\pdfoutput=1
\documentclass[journal,comsoc]{IEEEtran}
\usepackage[T1]{fontenc}
\usepackage[utf8]{inputenc}
\usepackage{amsmath,amssymb,amsfonts,amsthm}
\usepackage{mathrsfs}
\usepackage{enumitem}
\usepackage{graphicx}
\usepackage{textcomp}
\usepackage{xcolor}
\usepackage{booktabs}
\usepackage{url}
\usepackage[sort,compress]{cite}
\usepackage[ruled,vlined]{algorithm2e}
\usepackage[hidelinks]{hyperref}

\newtheorem{proposition}{Proposition}

\begin{document}
\bstctlcite{IEEEexample:BSTcontrol}

\title{Distributed Hybrid Quantum-Classical Optimization for Robust AoII Minimization in Satellite-Ground Integrated Edge Networks}
\author{Mohammad~Arif~Hossain,~\IEEEmembership{Member,~IEEE}, Tanzimul~Alam~Fahim,~Weiqi~Liu,~\IEEEmembership{Member,~IEEE}, and~Nirwan~Ansari,~\IEEEmembership{Life~Fellow,~IEEE}%
\thanks{This work has been submitted to the IEEE Transactions on Cognitive Communications and Networking (TCCN) for possible publication and is currently under review. Copyright may be transferred without notice, after which this version may no longer be accessible.}%
\thanks{Mohammad~Arif~Hossain is with the Department of Engineering Technology and the Quantum Research Interdisciplinary Science and Education (QRISE) Center, Middle Tennessee State University (MTSU), Murfreesboro, TN 37132 USA (e-mail: Mohammad.Hossain@mtsu.edu). Tanzimul~Alam~Fahim is with the Computational Science PhD Program, MTSU, Murfreesboro, TN 37132 USA. Weiqi Liu is with the Department of Computer Science and Computer Information Systems, Auburn University at Montgomery, Montgomery, AL 36117 USA. Nirwan Ansari is with the Department of Electrical and Computer Engineering, New Jersey Institute of Technology, Newark, NJ 07102 USA.}}

\markboth{Preprint under review at IEEE Transactions on Cognitive Communications and Networking}%
{Hossain \MakeLowercase{\textit{et al.}}: Distributed Hybrid Quantum-Classical Optimization for Robust AoII Minimization in SGINs}

\maketitle

\begin{abstract}
Edge servers acting on monitored processes need estimates that are correct, not just recent. In satellite-ground integrated networks (SGINs), handover and shadowing interrupt links long enough for the process to change state several times, motivating the age of incorrect information (AoII). We formulate VANGUARD, a robust problem that jointly sets update rates, association, and bandwidth to minimize the worst-case AoII across disruption scenarios, and we solve it with SENTINEL, a distributed hybrid quantum-classical procedure. A renewal-interval decomposition yields a closed-form cost per delivery interval for AoII. The minimax is dualized into classically updated scenario weights, so each domain solves a linear assignment encoded as a compact quadratic unconstrained binary optimization (QUBO) problem with XY-mixer-enforced one-hot constraints; domains are coordinated by the alternating direction method of multipliers (ADMM), and every schedule is certified by its exact worst-case AoII. Under Starlink passes and the 3GPP TR 38.811 channel, the certificate held in 98.7\% of epochs, and the robust one was 2 to 7\% below the nominal one at no cost in realized AoII. The compact encoding uses four to six times fewer qubits than an epigraph encoding; thirteen domains, each 8 to 38 qubits, run on IBM Heron processors and are decoded to the exact optimum.
\end{abstract}


\IEEEpeerreviewmaketitle

%=======================================================================
\section{Introduction}
\label{sec:intro}
%=======================================================================

Sixth-generation (6G) and beyond networks are expected to place decision-making at an edge server that monitors a hazard indicator, the integrity of a structure, or the position of an asset and must act on its own estimate of that process, often without a human in the loop~\cite{Saad2020}. What matters to such a server is whether its estimate is still correct. Terrestrial mobile edge computing provides capacity and low-delay connectivity where ground infrastructure exists, which is often not where monitoring is most needed, while satellite links extend coverage to the remaining regions at the cost of intermittent connectivity~\cite{Giordani2021}.
Satellite links fail in ways that terrestrial links rarely do. A low Earth orbit (LEO) satellite sets and hands the user over to the next one every few minutes; buildings and foliage block the line of sight for the length of a pass; shadow fading in the non-line-of-sight state exceeds ten decibels~\cite{TR38811}. During any of these interruptions, the monitored process continues to evolve, and by the time the link is restored, the edge estimate may be off by several states rather than simply outdated. The age of information (AoI)~\cite{Yates2021} measures only elapsed time and cannot distinguish between these two situations. The age of incorrect information (AoII)~\cite{Maatouk2023} can be zero even when the estimate is correct, and it grows with the duration and magnitude of a mismatch. AoII has been analyzed mainly for one source and one monitor on a point-to-point channel. Carrying it to a network raises two difficulties: the metric depends on the whole trajectory of the source since the estimate was last correct, so it does not fit the per-slot cost structure of the assignment problems used for satellite-ground integrated network (SGIN) resource allocation; and the duration of a mismatch is governed by outages of uncertain length, so an allocation must be judged by its worst case over plausible disruptions rather than by its average.

We treat the problem as robust combinatorial optimization. A renewal-interval decomposition shows that a Markov source with known transition matrix accumulates, between two deliveries, an expected AoII that depends only on the length of the interval and can be tabulated offline, so the channel enters the objective through a single number per scenario and the resulting robust scheduling problem, which we call \textbf{V}olatility-\textbf{A}ware \textbf{N}etwork-wide \textbf{G}eneration, \textbf{U}pdate, \textbf{A}ssociation, and \textbf{R}esource orchestration under \textbf{D}isruption (VANGUARD), is a minimax assignment problem over binary variables. The question is then how to present such a problem to a quantum sampler. A direct encoding of the minimax, with an epigraph register and a slack register per scenario on top of every admissible decision variable, needs four to six times the qubits of the encoding proposed here and produces a penalty landscape on which simulated annealing fails to find a feasible point; we document this in Section~\ref{sec:performance}. Our procedure, \textbf{S}atellite-ground \textbf{E}dge \textbf{N}etwork orchestration for \textbf{T}rustworthy \textbf{I}nformation via \textbf{N}ode-level \textbf{E}stimation \textbf{L}imits (SENTINEL), instead dualizes the minimax into scenario weights that a classical outer loop updates, so that each inner subproblem is linear and is encoded as a compact quadratic unconstrained binary optimization (QUBO) problem over the decision variables and a few capacity-slack bits; the one-hot constraint of each user is enforced by the XY mixer of a quantum approximate optimization algorithm (QAOA) circuit~\cite{Hadfield2019} rather than by a penalty. Satellite-ground domains are coordinated by the sharing form of the alternating direction method of multipliers (ADMM)~\cite{Boyd2011}, and every assembled schedule is evaluated exactly under every scenario, so the reported worst case is a certificate that does not rest on the solver's optimality. The procedure is hybrid at every level: a classical optimizer tunes the circuit angles; the scenario weights, the coordination, the repair, and the certificate are classical; and domains beyond the sampler's capacity are solved by a classical annealer on the same QUBO. Section~\ref{sec:performance} states which solver produced each result.

The evaluation uses Starlink ephemerides propagated with the simplified general perturbations (SGP4) model~\cite{Vallado2006}, the suburban channel of the Third Generation Partnership Project (3GPP) Technical Report (TR) 38.811~\cite{TR38811}, acknowledgement feedback for every policy, an exact integer-programming reference, common random episodes with paired tests, and runs of the quantum step on two IBM Heron processors. The worst-case certificate held in 98.7\% of orchestration epochs under real passes; the robust schedule's certificate is a median 2 to 7\% below its nominal counterpart's and lower in 67 to 100\% of seeds, while the two are statistically indistinguishable in realized AoII, so the value of the scenario set is the guarantee, not a lower average. A per-slot policy that transmits whenever its observed AoII is positive is the strongest heuristic when feedback is available and coincides with the Whittle-index scheduler~\cite{Kadota2018} at these bandwidths; SENTINEL with the same trigger comes within 10 to 20\% of it while staying within its schedule's certificate in 97.9 to 98.7\% of epochs, and is 25 to 30\% better than it at 15 to 20 users without feedback. The compact formulation matches the exact optimum on the instances with 10 and 25 user equipments (UEs) with four to six times fewer qubits than the epigraph encoding, and thirteen domains of 8 to 38 qubits sampled on \texttt{ibm\_kingston} and \texttt{ibm\_fez} decoded to the exact optimum. On every instance that fits a current device greedy association is already optimal, so the quantum step is a correct, hardware-runnable constrained sampler and not yet a better optimizer; we quantify this per shot against a uniform-sampling baseline and state the scale at which it would change. The contributions of this paper are the following.
\begin{itemize}[leftmargin=13pt]
\item The robust scheduling problem VANGUARD, which minimizes the worst-case time-average AoII of a network of monitored edge nodes over a set of satellite shadowing and handover scenarios, jointly over update rates, association, and bandwidth, under capacity, visibility, rate, and AoII-budget constraints.
\item A renewal-interval decomposition that expresses the expected AoII accumulated between deliveries in closed form and reconciles the history dependence of AoII with binary-variable optimization; verified against Monte Carlo simulation to within one percent.
\item SENTINEL, a distributed hybrid quantum-classical procedure that dualizes the minimax into classically updated scenario weights, encodes each domain's subproblem as a compact QUBO with XY-mixer-enforced one-hot structure, prunes dominated candidates, and coordinates the domains by ADMM with a warm start and a cost-scaled penalty, with a certificate of the worst-case AoII of every schedule.
\item An evaluation under Starlink ephemerides and the TR 38.811 channel against an exact reference, with paired statistics, Whittle-index and trained learning baselines, an ablation of failure feedback, sensitivity studies, and runs of the quantum step on two superconducting processors with its scale limit stated explicitly.
\end{itemize}

Section~\ref{sec:sota} reviews related work. Section~\ref{sec:model} presents the system model, Section~\ref{sec:formulation} formulates VANGUARD and its binary reformulations, Section~\ref{sec:solution} develops the dualized compact QUBO, the XY-mixer QAOA, and the ADMM coordination, Section~\ref{sec:performance} reports the evaluation, and Section~\ref{sec:conclusion} concludes.

%=======================================================================
%=======================================================================
\section{State of the Art}
\label{sec:sota}
%=======================================================================

\subsubsection{Satellite-Ground Integrated Edge Computing} Mobile edge computing has been extended beyond terrestrial coverage by placing edge servers on satellites and base stations~\cite{Kato2019}, with task offloading under deterministic or average channel assumptions as the dominant problem~\cite{Liu_2026} and completion delay as the objective; the correctness of what the server believes about the monitored process has rarely been the quantity optimized: freshness has been addressed through the peak AoI of the LEO satellite Internet of Things (IoT)~\cite{Liao2024} and jointly with handover in SAGINs~\cite{Lang2025}, and AoII has been minimized for a satellite IoT system with non-orthogonal multiple access~\cite{Hong2024}. The 3GPP study on non-terrestrial networks (NTNs)~\cite{TR38811} provides elevation-dependent line-of-sight probabilities, clutter losses, and shadow-fading statistics, and SGP4 with public two-line element (TLE) sets~\cite{Vallado2006} gives pass geometry for any constellation; most SGIN evaluations use neither, although real passes are disrupted mainly by handovers and non-line-of-sight passes, not coverage gaps~\cite{hossain2026satlock}.

\subsubsection{Age of Information and Age of Incorrect Information} AoI is a widely studied but content-agnostic timeliness measure~\cite{Yates2021}. AoII~\cite{Maatouk2020} penalizes the duration and magnitude of estimation error and has been analyzed for a single source under threshold and semi-Markov policies~\cite{Maatouk2023} and over unreliable, power-constrained channels~\cite{Chen2021}. For several sources over shared channels the standard tool is the Whittle index of the restless bandit~\cite{Whittle1988}, applied to AoI by Kadota et al.~\cite{Kadota2018} and to AoII by Kriouile and Assaad~\cite{Kriouile2021} and Ayik et al.~\cite{Ayik2023}; we include it as a baseline. Missing is a formulation in which association, bandwidth, and update rates are chosen jointly across heterogeneous nodes, under outages spanning several source transitions, with a worst-case guarantee.

\subsubsection{Robust Allocation and Learning-Based Scheduling} Robust optimization minimizes cost against the worst realization of an uncertain parameter in a bounded set~\cite{BenTal2009}. When the uncertainty set is a finite list of scenarios, the minimax is a maximization over scenario weights in the simplex, and the saddle point can be approached by multiplicative-weights updates~\cite{Arora2012}; this is the mechanism we use to keep the quantum subproblem linear. Robust allocation in satellite-ground edge computing has treated uncertainty as a perturbation of rate or delay~\cite{Chen2021DRO}, not as the driver of estimation error. Deep reinforcement learning~\cite{Schulman2017} has been applied to AoI- and AoII-aware scheduling, including in SAGINs~\cite{Lang2025}; it optimizes an expected return and gives no bound on the estimation error under disruptions outside its training distribution. That does not disqualify it as a baseline, and we train one with the same reward, feedback, and episodes as every other policy.

\subsubsection{Quantum-Assisted Combinatorial Optimization} QUBO formulations~\cite{Lucas2014} solved by QAOA~\cite{Farhi2014} or by annealing~\cite{Albash2018} have been proposed for wireless problems including detection in large multiple-input multiple-output systems and multi-beam satellite scheduling~\cite{Yan2026}. Two developments bear on constrained assignment: the alternating-operator form of QAOA~\cite{Hadfield2019}, whose XY mixers~\cite{Wang2020} preserve a constraint subspace so that one-hot constraints need no penalty terms, and the concentration of optimal QAOA parameters across instances~\cite{Brandao2018}, which allows angles optimized on a small instance to be reused on a larger one; error suppression by dynamical decoupling~\cite{Viola1999} and randomized compiling~\cite{Wallman2016} is inexpensive, and we report its effect. The closest algorithmic precedents are ADMM splitting with QUBO subproblems~\cite{Gambella2020}, heuristic robust variants of annealing and QAOA~\cite{Halffmann2024}, and warm-started XY mixers run on 144-qubit hardware instances with classical repair of infeasible samples~\cite{Bucher2026}. QAOA has not been applied to a correctness objective, and reports of its use for network resource allocation rarely state at what size the quantum sampler stops being competitive with a classical heuristic.

\subsubsection{Research Gap} No existing work formulates AoII as a robust minimax multi-node scheduling problem with joint association and resource allocation, targets information correctness rather than delay in robust SGIN allocation, or applies quantum-assisted optimization to a correctness objective with an encoding a current device can execute. VANGUARD and SENTINEL address these three points.
%=======================================================================
\section{System Model}
\label{sec:model}
%=======================================================================

\begin{figure}[!t]\centering
\includegraphics[width=\columnwidth]{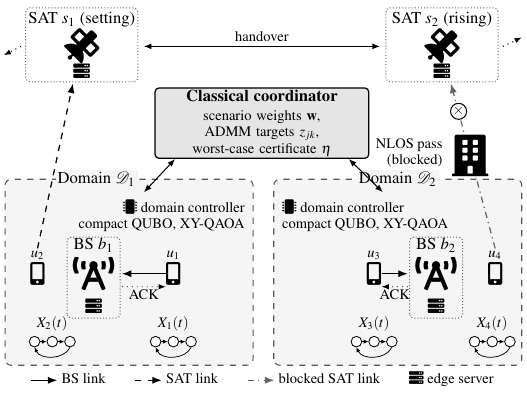}
\caption{System model: two satellite-ground domains with their base stations (BSs), shared LEO satellites, monitored sources $X_i(t)$, and the classical coordinator. Solid: BS link with acknowledgement (dotted); dashed: satellite link; dash-dotted: satellite link blocked during a non-line-of-sight pass.}
\label{Fig1}
\end{figure}

We consider the satellite-ground integrated edge network of Fig.~\ref{Fig1}, in which monitoring UEs observe time-varying physical processes and transmit status updates to edge servers hosted at terrestrial base stations (BSs) or at visible LEO satellites (SATs). Let $\mathscr{U}$, $\mathscr{S}$, and $\mathscr{B}$ denote the sets of UEs, SATs, and BSs, indexed by $u_i$, $s_j$, and $b_n$, and let $\mathscr{M}=\mathscr{S}\cup\mathscr{B}$ be the set of serving nodes. Each serving node hosts an edge server that keeps a running estimate of every process assigned to it and acts on that estimate. Time is slotted with slot duration $T_s$. UEs have low mobility over a decision interval; the satellite topology changes continuously because of orbital motion. Table~\ref{tab:notation} lists the notation.

In contrast to the task-offloading models that dominate SGIN edge computing, the quantity of interest here is the correctness of the server's estimate. It depends on the dynamics of the monitored process, on when the UE samples and transmits, and on whether the transmission is delivered before the process moves again. Satellite handover and shadowing make the last of these uncertain, and an allocation chosen under nominal channel conditions can leave the server with an estimate that is wrong rather than late.

The network is partitioned into $K$ local satellite-ground domains $\mathcal{D}_k=(\mathscr{U}_k,\mathscr{S}_k, \mathscr{B}_k)$, $k\in\mathscr{K}=\{1,\ldots,K\}$, where the UE subsets partition $\mathscr{U}$ and a SAT may be visible to several domains. Each domain has a classical controller that builds and solves a local robust binary problem with a hybrid QAOA solver; a classical coordinator reconciles the bandwidth of SATs shared across domains and certifies every assembled network schedule by its exact worst-case AoII.

\begin{table}[!t]
\centering
\caption{Key Notations}
\label{tab:notation}
\scriptsize
\renewcommand{\arraystretch}{1.05}
\begin{tabular}{|p{0.27\columnwidth}|p{0.61\columnwidth}|}
\hline
\textbf{Symbol} & \textbf{Definition} \\
\hline
$\mathscr{U}$, $\mathscr{S}$, $\mathscr{B}$ & Sets of UEs, SATs, and BSs \\
\hline
$\alpha_{ij}$, $\omega_{in}$ & Association of $u_i$ with $s_j$, $b_n$ \\
\hline
$\beta^{us}_{ij}$, $\beta^{ub}_{in}$ & Bandwidth allocated to $u_i$ by $s_j$, $b_n$ \\
\hline
$B^{s}_{j}$, $B^{b}_{n}$ & Bandwidth capacity of $s_j$, $b_n$ \\
\hline
$R^{us}_{ij}(\ell)$, $R^{ub}_{in}$ & Uplink rate of $u_i$ to $s_j$ (scenario $\ell$), to $b_n$ \\
\hline
$V_{ij}(t)$, $\Theta$ & Elevation of $s_j$ at $u_i$; minimum elevation \\
\hline
$\Delta_i^{II}(t)$, $\Delta_i^{II,\max}$ & Instantaneous AoII; AoII budget \\
\hline
$\Psi_i(Y)$ & Expected AoII accumulated over an interval of length $Y$ \\
\hline
$\rho_i$, $\bar\rho_i^{h}$, $\bar\beta_m^{q}$ & Update rate of $u_i$; $h$-th rate level ($H_i$ levels); $q$-th bandwidth level of node $m$ ($Q_m$ levels) \\
\hline
$x_{imqh}$ & $1$ if $u_i$ uses node $m$, bandwidth level $q$, rate level $h$ \\
\hline
$\mathscr{A}_i$, $\mathscr{C}_i$, $\varkappa$ & Admissible combinations of $u_i$; pruned candidates; candidates kept per UE \\
\hline
$\psi^{(\ell)}_{imqh}$ & Time-average AoII cost under scenario $\ell$ \\
\hline
$\mathbf{w}$, $\mu$, $\eta$ & Scenario weights; weight-update step; worst-case AoII \\
\hline
$z_{jk}$, $\upsilon_{jk}$, $\lambda_{\mathrm{c}}$ & ADMM bandwidth target, scaled dual, penalty \\
\hline
\end{tabular}
\end{table}

\subsection{Satellite Geometry and Communications Model}
\label{sec:geometry}

Satellite positions are obtained from a constellation ephemeris; the evaluation propagates a public two-line element set with SGP4~\cite{Vallado2006}, but the model only requires, for every SAT $s_j$ and UE $u_i$, the slant range $d^{us}_{ij}(t)$ and the elevation $V_{ij}(t)$ in each slot. A SAT can serve $u_i$ in slot $t$ only if $V_{ij}(t)\geq\Theta$. A UE keeps its serving SAT until it sets below $\Theta$ and is then handed over to the SAT with the highest elevation, so a dense constellation yields passes of several minutes separated by handover interruptions, with no coverage gaps at the elevation masks considered here.

The nominal path loss (in dB) between $u_i$ and $s_j$ is
\begin{equation}
    \varphi^{us}_{ij}=\varsigma^{us}+\zeta^{us}\log_{10}\!\left(f_{c}^{us}\right)
    +\xi^{us}\log_{10}\!\left(d^{us}_{ij}\right)+\lambda^{\mathrm{atm}}(V_{ij}),
\end{equation}
with carrier frequency $f_c^{us}$, the reference loss and exponents $\varsigma^{us}$, $\zeta^{us}$, $\xi^{us}$ of free-space propagation, and the atmospheric absorption $\lambda^{\mathrm{atm}}(V)\propto 1/\sin V$. The UE-BS path loss $\varphi^{ub}_{in}$ uses the ground distance $d^{ub}_{in}$, a terrestrial exponent $\xi^{ub}>2$, and a fixed margin. With transmit power $p_i$ and antenna gains $G^{u}$ and $G^{s}$ at the terminal and the satellite, the nominal uplink rate with bandwidth $\beta^{us}_{ij}$ is
\begin{equation}
    R^{us}_{ij}=\beta_{ij}^{us}\,\min\!\left\{\log_{2}\!\left(1+
    \frac{10^{(p_{i}+G^{u}+G^{s}-\varphi^{us}_{ij})/10}}{N_0\,\beta^{us}_{ij}\,F^{s}}\right),\ \bar{\gamma}\right\},
    \label{eq:rate_us}
\end{equation}
where $N_0$ is the noise power spectral density, $F^{s}$ the receiver noise figure (NF), and $\bar\gamma$ the spectral efficiency (SE) of the highest modulation and coding scheme; $R^{ub}_{in}$ is defined in the same way with $G^{b}$. Each node divides its bandwidth among the UEs it serves, so a link's rate does not depend on other UEs' association once the bandwidth level is fixed. Co-channel interference is absorbed in the noise figure, and residual Doppler, which the NTN frequency pre-compensation removes to within the subcarrier spacing, is not modeled.

\emph{Large-scale channel.} Each satellite link also carries a large-scale state drawn once per pass from the TR 38.811 model~\cite{TR38811}: with the elevation-dependent line-of-sight (LOS) probability $p_{\mathrm{LOS}}(V)$ the pass is in LOS; otherwise it is in non-line-of-sight (NLOS) and carries the clutter loss $\lambda^{\mathrm{cl}}(V)$ of the band (17 to 25~dB in the suburban tables, with Ku-band values interpolated between the S-band and Ka-band entries). Shadow fading with the state-dependent deviation $\sigma_{\mathrm{sf}}(V)$ is drawn per pass and small-scale fading per slot. About one pass segment in twenty is NLOS in the suburban environment of Section~\ref{sec:performance}, and such a pass blocks the link throughout at any bandwidth.

\subsection{Monitored Source and Estimation Model}
\label{sec:source}

Each UE $u_i$ observes a process $X_i(t)$ modeled as a discrete-time Markov chain on a finite state space $\mathcal{X}_i$ with transition matrix $P_i$, in which the states are ordered severity or status levels. The UE controls a generation indicator $g_i(t)\in\{0,1\}$ that specifies whether an update is sampled and transmitted in slot $t$. We assume that $T_s$ exceeds the transmission and propagation delay of a successful update, so that an update sampled at the start of slot $t$ is delivered within that slot or lost. Let $t_{i,\kappa}$ be the $\kappa$-th slot in which an update from $u_i$ is delivered and $Y_{i,\kappa}=t_{i,\kappa+1}-t_{i,\kappa}$ the inter-delivery interval. The server assigned to $u_i$ holds the estimate
\begin{equation}
    \hat{X}_i(t)=X_i(t_{i,\kappa}),\qquad t_{i,\kappa}\leq t<t_{i,\kappa+1},
\end{equation}
which stays fixed between deliveries while the process moves. With $\Lambda_i(t)=\max\{\tau\leq t: X_i(\tau)=\hat{X}_i(\tau)\}$ the most recent slot with a correct estimate, the instantaneous AoII is~\cite{Maatouk2020}
\begin{equation}
    \Delta_i^{II}(t)=\bigl(t-\Lambda_i(t)\bigr)\,f_i\!\left(X_i(t),\hat{X}_i(t)\right),
    \label{eq:aoii_def}
\end{equation}
where $f_i(x,\hat{x})\geq 0$ with $f_i(x,x)=0$ measures mismatch severity; we use $f_i(x,\hat{x})=|x-\hat{x}|$. AoII returns to zero without a delivery when the process revisits the estimated state.

\subsection{Disruption Scenarios}
\label{sec:scenarios}

The orchestrator plans against a finite scenario set $\mathscr{L}=\{1,\ldots,L\}$. Scenario $\ell$ adds a loss $\delta^{us}_{ij}(\ell)$ to every UE-SAT link for a window of $T^{\mathrm{dis}}(\ell)$ slots,
\begin{equation}
    \tilde{\varphi}^{us}_{ij}(\ell)=\varphi^{us}_{ij}+\delta^{us}_{ij}(\ell),
    \label{eq:realized_pathloss}
\end{equation}
the scenario rate $R^{us}_{ij}(\ell)$ follows from~\eqref{eq:rate_us}, and an update sampled in slot $t$ is delivered under scenario $\ell$ only if $R^{us}_{ij}(\ell)\,T_s\geq D_i$ and $V_{ij}(t)\geq\Theta$, where $D_i$ is the update size in bits. BS links are unaffected. We use five scenarios: the nominal channel; a deep shadowing event in which every satellite link suffers the NLOS shadow-fading deviation $\sigma_{\mathrm{sf}}$ (12~dB in the suburban tables) for a window of $T^{\mathrm{dis}}$ slots; the handover interruption of each individual satellite, a complete outage of that link for the duration of a handover; and the joint outage of all satellites for the deep-shadowing window. A mild 6~dB scenario was dropped after the sensitivity study of Section~\ref{sec:certificate} showed it never binds; the handover scenarios are kept per satellite because the robust hedge appears only with the second shared satellite's handover.

A scenario thus determines the length of the outage window and therefore the inter-delivery interval $Y_{i,\kappa}(\ell)$, the number of slots from $t_{i,\kappa}$ until the next slot in which $g_i(t)=1$ and the delivery condition holds; a longer window lets $X_i(t)$ pass through more states. The scenario set is the planner's model of the disruption process, while the simulator of Section~\ref{sec:performance} draws its disruptions from the TR 38.811 state process and the ephemeris; certificate validity measures their agreement.

\subsection{Renewal-Interval AoII Decomposition}
\label{sec:renewal}

Because $\Delta_i^{II}(t)$ depends on the trajectory of $X_i(t)$ since the estimate was last correct, it is not a per-slot function of the channel state and cannot be used directly as the cost coefficient of a binary variable. The delivery process is, however, a renewal process. Under the same-slot delivery assumption the estimate is exact at each delivery, so $\Delta_i^{II}(t_{i,\kappa})=0$, and between two deliveries the estimate is fixed at $\hat{x}=X_i(t_{i,\kappa})$ while the AoII evolves under the source dynamics alone.

Let $\tilde{P}_{i,\hat{x}}=P_i\bigl(\mathbf{I}-\mathbf{e}_{\hat{x}}\mathbf{e}_{\hat{x}}^{\top}\bigr)$ be the taboo matrix with the $\hat{x}$-th column of $P_i$ set to zero, so that $[\tilde{P}_{i,\hat{x}}^{\iota}]_{\hat{x}y}$ is the probability of moving from $\hat{x}$ to $y$ in $\iota$ steps without returning to $\hat{x}$, and let $\mathbf{f}_{i,\hat{x}}$ be the vector with entries $f_i(y,\hat{x})$. The AoII $s$ slots after delivery equals $\iota\,f_i(y,\hat{x})$ exactly when the source was last at $\hat{x}$ at slot $s-\iota$ and since moved to $y$ without returning, so
\begin{equation}
    \mathbb{E}\!\left[\Delta_i^{II}(t_{i,\kappa}+s)\,\middle|\,\hat{x}\right]
    =\sum_{\iota=1}^{s} \iota\,\bigl[P_i^{\,s-\iota}\bigr]_{\hat{x}\hat{x}}
    \bigl[\tilde{P}_{i,\hat{x}}^{\,\iota}\,\mathbf{f}_{i,\hat{x}}\bigr]_{\hat{x}},
    \label{eq:aoii_slot}
\end{equation}
and the expected AoII accumulated over an interval of length $Y$ is
\begin{align}
    \Psi_i(Y\mid\hat{x})&=\sum_{s=1}^{Y-1}
    \mathbb{E}\!\left[\Delta_i^{II}(t_{i,\kappa}+s)\,\middle|\,\hat{x}\right],
    \nonumber\\
    \Psi_i(Y)&=\sum_{\hat{x}\in\mathcal{X}_i}\pi_i(\hat{x})\,\Psi_i(Y\mid\hat{x}),
    \label{eq:psi_def}
\end{align}
where $\pi_i$ is the stationary distribution of $P_i$. Averaging over $\pi_i$ is valid because the update decisions do not depend on the realized source state, so the delivered state is stationary. By the renewal-reward theorem the time-average AoII of $u_i$ when its deliveries are $Y$ slots apart is $\Psi_i(Y)/Y$. Both quantities depend only on $Y$ and $P_i$ and are tabulated once per UE for $Y\leq Y^{\max}$; the table agrees with Monte Carlo simulation of the source and delivery process over $1.5\times 10^{6}$ slots within one percent at every interval length used.

Equation~\eqref{eq:psi_def} separates the source from the channel. The only channel-dependent quantity the planner needs is the inter-delivery interval, which under scenario $\ell$ is deterministic for a given allocation and is denoted $Y_i(\ell)$; the cost is the tabulated ratio $\Psi_i(Y_i(\ell))/Y_i(\ell)$.

\subsection{Binary Decision Model}
\label{sec:binary}

The update rates and bandwidths are continuous. We discretize the bandwidth of each node $m\in\mathscr{M}$ into $Q_m$ levels $\{\bar\beta_m^1,\ldots,\bar\beta_m^{Q_m}\}$ and the update rate of $u_i$ into $H_i$ levels $\{\bar\rho_i^1,\ldots,\bar\rho_i^{H_i}\}\subset[\rho_i^{\min},\rho_i^{\max}]$, and define
\begin{equation}
    x_{imqh}=
    \begin{cases}
        1, & \text{if } u_i \text{ uses node } m\text{, level } q\text{, rate } h, \\
        0, & \text{otherwise}.
    \end{cases}
    \label{eq:x_def}
\end{equation}
The original decisions are recovered as $\alpha_{ij}=\sum_{q,h}x_{ijqh}$, $\omega_{in}=\sum_{q,h}x_{inqh}$, $\beta^{us}_{ij}=\sum_{q,h}\bar\beta_j^q\,x_{ijqh}$, $\beta^{ub}_{in}=\sum_{q,h}\bar\beta_n^q\,x_{inqh}$, and $\rho_i=\sum_{m,q,h}\bar\rho_i^h\,x_{imqh}$. The scenario does not index the variable, because it is realized by the environment rather than chosen; each variable instead carries $L$ precomputed cost coefficients. After the inadmissible combinations are removed (Section~\ref{sec:reform}) the number of decision variables is
\begin{equation}
    N_{x}=\sum_{i\in\mathscr{U}}|\mathscr{A}_i|\;\leq\;\sum_{i\in\mathscr{U}}\sum_{m\in\mathscr{M}}Q_{m}H_{i},
    \label{eq:Nx}
\end{equation}
which grows linearly with the number of UEs, nodes, and levels. Section~\ref{sec:pruning} prunes each $\mathscr{A}_i$ further to a candidate set $\mathscr{C}_i$ of the $\varkappa$ combinations with the smallest worst-case cost, which in our instances never changed the optimum and reduced the qubit count by a further factor of two to three.

\subsection{Delivery Feedback and Link-Failure Handling}
\label{sec:feedback}

A serving node acknowledges each delivered update. A UE whose transmissions to a node have failed in every slot of a window of $T^{\mathrm{f}}$ slots declares the link failed and avoids that node for $T^{\mathrm{b}}$ slots or, for a satellite, until its next handover; its controller re-solves the domain without that link at the next orchestration interval. The rule applies to every policy in Section~\ref{sec:performance}, as any state-aware baseline needs the same information to compete; results with and without it are reported.

\subsection{Distributed Orchestration Model}
\label{sec:orchestration}

Each domain $\mathcal{D}_k$ holds the local decision vector $\mathbf{x}_k=\{x_{imqh}\mid u_i\in\mathscr{U}_k,\, (m,q,h)\in\mathscr{C}_i\}$ and the local state $\mathbf{s}_k=[\boldsymbol{\Psi}_k,\mathbf{V}_k,\mathbf{X}_k, \mathbf{B}_k,\mathbf{E}_k,\boldsymbol{\Xi}_k]$, with $\boldsymbol{\Psi}_k$ the tabulated AoII costs of its UEs, $\mathbf{V}_k$ the elevations, $\mathbf{X}_k$ the transition matrices and budgets, $\mathbf{B}_k$ the bandwidth levels, $\mathbf{E}_k$ the ephemeris over the decision interval, and $\boldsymbol{\Xi}_k$ the scenario set restricted to the domain. The controller re-solves whenever a satellite becomes visible or sets, when a link failure is reported, and at least once per re-solve period otherwise. Within a solve, a classical outer loop maintains scenario weights, each domain solves a linear assignment subproblem on a quantum sampler (or, when the domain exceeds the sampler's capacity, on a classical annealer with the same QUBO), and the coordinator runs the sharing form of ADMM on the shared satellites' bandwidth and evaluates the worst-case AoII of each assembled schedule exactly. An orchestration latency of $\epsilon_i$ slots extends every inter-delivery interval and enters the cost through $\Psi_i(Y_i(\ell)+\epsilon_i)$. The evaluation sets $\epsilon_i=0$ and reports the measured solver time separately: visibility changes are known from the ephemeris, so an epoch's schedule can be computed in advance.

%=======================================================================
\section{Problem Formulation}
\label{sec:formulation}
%=======================================================================

The objective is to minimize the worst-case time-average AoII across all monitored UEs over the scenario set, by jointly choosing update rates, association, and bandwidth. Let $\boldsymbol{\chi}_i=\bigl(\rho_i,\{\alpha_{ij}\}_{j}, \{\omega_{in}\}_{n},\{\beta^{us}_{ij}\}_{j},\{\beta^{ub}_{in}\}_{n}\bigr)$ collect the decisions of $u_i$ and $Y_i(\boldsymbol{\chi}_i,\boldsymbol{\xi})$ its inter-delivery interval under a disruption realization $\boldsymbol{\xi}$ from the uncertainty set $\Omega$, which Section~\ref{sec:scenarios} represents by the finite scenario set $\mathscr{L}$. The robust VANGUARD problem is
\begin{align}
    \mathscr{P}_{1}:\;&\min_{\{\boldsymbol{\chi}_i\}_{i\in\mathscr{U}}}
    \;\max_{\boldsymbol{\xi}\in\Omega}
    \;\frac{1}{|\mathscr{U}|}\sum_{i\in\mathscr{U}}
    \frac{\Psi_{i}\bigl(Y_{i}(\boldsymbol{\chi}_i,\boldsymbol{\xi})\bigr)}
    {Y_{i}(\boldsymbol{\chi}_i,\boldsymbol{\xi})}
    \nonumber\\
    \text{s.t.}\;\;
    &\text{C1}:\sum_{j\in\mathscr{S}}\alpha_{ij}+\sum_{n\in\mathscr{B}}\omega_{in}=1,\;\forall\,i\in\mathscr{U},
    \nonumber\\
    &\text{C2}:\sum_{i\in\mathscr{U}}\alpha_{ij}\,\beta^{us}_{ij}\leq B^{s}_{j},\;\forall\,j\in\mathscr{S},
    \nonumber\\
    &\text{C3}:\sum_{i\in\mathscr{U}}\omega_{in}\,\beta^{ub}_{in}\leq B^{b}_{n},\;\forall\,n\in\mathscr{B},
    \nonumber\\
    &\text{C4}:\frac{\Psi_{i}\bigl(Y_{i}(\boldsymbol{\chi}_i,\boldsymbol{\xi})\bigr)}
    {Y_{i}(\boldsymbol{\chi}_i,\boldsymbol{\xi})}\leq\Delta_{i}^{II,\max},
    \;\forall\,i\in\mathscr{U},\;\forall\,\boldsymbol{\xi}\in\Omega,
    \nonumber\\
    &\text{C5}:\alpha_{ij}\bigl(V_{ij}-\Theta\bigr)\geq 0,\;\forall\,i\in\mathscr{U},\;\forall\,j\in\mathscr{S},
    \nonumber\\
    &\text{C6}:\rho_{i}^{\min}\leq\rho_{i}\leq\rho_{i}^{\max},\;\forall\,i\in\mathscr{U},
    \nonumber\\
    &\text{C7}:\alpha_{ij},\,\omega_{in}\in\{0,1\},\;\forall\,i,j,n.
    \nonumber
\end{align}
C1 routes each UE through exactly one serving node per decision interval. C2 and C3 are the bandwidth capacities. C4 is the correctness constraint: the expected time-average AoII of every UE stays within its budget $\Delta_i^{II,\max}$ under every realization, including the worst. It is stated in expectation because the instantaneous AoII of a Markov source cannot be bounded deterministically; the source may stay in a mismatched state for a whole interval. C5 restricts association to visible satellites, C6 bounds the update rate by the freshness requirement and the UE's energy or channel-access budget, and C7 makes association binary.

$\mathscr{P}_{1}$ is a mixed-integer nonlinear minimax program. With a single scenario and fixed rates it contains the generalized assignment problem and is NP-hard. The inner maximization adds a difficulty of its own: because $\Psi_i$ is a nonlinear, source-specific function of the interval, the realization that maximizes the network AoII is not necessarily the one with the worst channel but depends on how each source's volatility interacts with the outage it induces; and the association set defined by C5 changes at every decision interval as satellites move.

\subsection{Robust Binary Reformulation}
\label{sec:reform}

With the discretization of Section~\ref{sec:binary} and the variable $x_{imqh}$ of~\eqref{eq:x_def}, each scenario $\ell$ gives, for every UE-node pair and bandwidth level, a rate $R_{imq}(\ell)$ from~\eqref{eq:rate_us} with the scenario path loss. The link is in outage throughout the window if that rate cannot deliver an update within a slot, so the outage is
\begin{equation}
    \tau_{imq}(\ell)=
    \begin{cases}
        T^{\mathrm{dis}}(\ell)\,\mathbb{I}\bigl\{R_{imq}(\ell)\,T_s<D_i\bigr\}, & m\in\mathscr{S},\\
        0, & m\in\mathscr{B}.
    \end{cases}
    \label{eq:tau_out}
\end{equation}
The bandwidth level matters: a larger allocation can keep the rate above the delivery threshold through deep shadowing, whereas a handover interruption is an outage at any bandwidth. A UE whose update fails during an outage transmits a fresh sample as soon as the link returns. The worst case is an outage that begins at a scheduled sampling slot, which gives
\begin{equation}
    Y_{imqh}(\ell)=\left\lceil\frac{1}{\bar\rho_i^{h}T_s}\right\rceil+\tau_{imq}(\ell),\qquad
    \psi_{imqh}^{(\ell)}=\frac{\Psi_{i}\bigl(Y_{imqh}(\ell)\bigr)}{Y_{imqh}(\ell)},
    \label{eq:psi_imqh}
\end{equation}
precomputed for every $(i,m,q,h,\ell)$; with orchestration overhead, $Y_{imqh}(\ell)+\epsilon_i$ replaces $Y_{imqh}(\ell)$.

\emph{Pre-elimination.} The admissible set $\mathscr{A}_i$ consists of the combinations $(m,q,h)$ for which (i) $m\in\mathscr{B}$, or $m\in\mathscr{S}$ with $V_{im}\geq\Theta$; (ii) the nominal rate satisfies $R_{imq}T_s\geq D_i$; and (iii) $\max_{\ell}\psi^{(\ell)}_{imqh}\leq\Delta_i^{II,\max}$. Condition (i) enforces C5 and (iii) is the discrete form of C4: as each UE selects exactly one combination, it meets its budget under every scenario if and only if its selection lies in $\mathscr{A}_i$, so the budget is enforced by removing variables, not by penalty terms, at no qubit cost. A UE with empty $\mathscr{A}_i$ (a satellite-only UE with a budget below the AoII of a full handover outage) gets its combination of smallest worst-case cost and is reported as a budget violation.

The binary reformulation of VANGUARD is
\begin{align}
    \mathscr{P}_{2}:\;&\min_{\mathbf{x}}\;\max_{\ell\in\mathscr{L}}\;
    \frac{1}{|\mathscr{U}|}\sum_{i\in\mathscr{U}}\sum_{(m,q,h)\in\mathscr{A}_i}\psi_{imqh}^{(\ell)}\,x_{imqh}
    \nonumber\\
    \text{s.t.}\;\;
    &\text{D1}:\sum_{(m,q,h)\in\mathscr{A}_i}x_{imqh}=1,\;\forall\,i\in\mathscr{U},
    \nonumber\\
    &\text{D2}:\sum_{i\in\mathscr{U}}\sum_{(j,q,h)\in\mathscr{A}_i}\bar{\beta}^{q}_{j}\,x_{ijqh}\leq B^{s}_{j},
    \;\forall\,j\in\mathscr{S},
    \nonumber\\
    &\text{D3}:\sum_{i\in\mathscr{U}}\sum_{(n,q,h)\in\mathscr{A}_i}\bar{\beta}^{q}_{n}\,x_{inqh}\leq B^{b}_{n},
    \;\forall\,n\in\mathscr{B},
    \nonumber\\
    &\text{D4}:x_{imqh}\in\{0,1\},\;\forall\,i\in\mathscr{U},\;(m,q,h)\in\mathscr{A}_i.
    \nonumber
\end{align}
D1 is the one-hot selection, D2 and D3 the discrete capacities.

\subsection{Epigraph Form, Certificate, and Exact Reference}
\label{sec:epigraph}

Introducing the worst-case value $\eta$ gives the equivalent epigraph form
\begin{align}
    \mathscr{P}_{3}:\;&\min_{\mathbf{x},\,\eta}\;\eta
    \nonumber\\
    \text{s.t.}\;\;&\text{E1}:\frac{1}{|\mathscr{U}|}\sum_{i\in\mathscr{U}}\sum_{(m,q,h)\in\mathscr{A}_i}
    \psi_{imqh}^{(\ell)}\,x_{imqh}\leq\eta,\;\forall\,\ell\in\mathscr{L},
    \nonumber\\
    &\text{D1--D4},
    \nonumber
\end{align}
in which E1 is tight for at least one scenario at the optimum. $\mathscr{P}_{3}$ serves two purposes in this paper, neither of which is to be solved on a quantum device. First, it is the certificate: for any schedule $\mathbf{x}$ returned by any solver, the $L$ left-hand sides of E1 are evaluated exactly, so the worst-case time-average AoII of that schedule over $\mathscr{L}$ is known with certainty. Second, $\mathscr{P}_{3}$ is a mixed-integer linear program (MILP), and up to 75 UEs it is solved to proven optimality by a branch-and-cut solver~\cite{HiGHS2018} in well under a second, and only some 100-UE instances reach the time limit; this gives an exact reference for every heuristic and quantum solver we compare. A direct QUBO encoding of $\mathscr{P}_{3}$, with $\eta$ on a binary grid and a slack register for each of the $L$ constraints E1, is the natural alternative and serves as our baseline; it multiplies the qubit count by four to six and produces a penalty landscape in which simulated annealing found no feasible assignment on instances of ten UEs. Section~\ref{sec:solution} therefore works with $\mathscr{P}_{2}$ directly, replacing the inner maximization by scenario weights so that every quantum subproblem is linear in $\mathbf{x}$ and the only auxiliary qubits are the slack bits of the capacity constraints D2 and D3.
%=======================================================================
\section{Proposed Solution}
\label{sec:solution}
%=======================================================================

Problem $\mathscr{P}_{2}$ is a robust assignment problem whose only coupling across UEs is the capacity of the serving nodes. SENTINEL solves it in three steps. The minimax over scenarios is handled by a classical outer loop on scenario weights, so that every subproblem handed to a quantum sampler has a linear objective. Each domain's subproblem is encoded as a compact QUBO over its decision variables and a few capacity-slack bits, with the one-hot constraint enforced by the mixer of the quantum circuit rather than by a penalty. The domains are coordinated through the sharing form of ADMM. Every schedule produced at any point of the procedure is certified exactly by its worst-case AoII over $\mathscr{L}$, so it can be stopped at any time with a valid guarantee. The procedure is hybrid at every level, as Section~\ref{sec:intro} describes, and Section~\ref{sec:performance} states which inner solver produced each result.

%-----------------------------------------------------------------------
\subsection{Scenario-Weight Dualization of the Minimax}
\label{sec:dual}
%-----------------------------------------------------------------------

Write $c^{(\ell)}(\mathbf{x})=\frac{1}{|\mathscr{U}|}\sum_{i}\sum_{(m,q,h)\in\mathscr{A}_i}\psi^{(\ell)}_{imqh}x_{imqh}$ for the network AoII under scenario $\ell$. For any $\mathbf{x}$, with $\Delta_L=\{\mathbf{w}\geq 0,\ \sum_\ell w_\ell=1\}$ the probability simplex of scenario weights,
\begin{equation}
    \max_{\ell\in\mathscr{L}} c^{(\ell)}(\mathbf{x})
    =\max_{\mathbf{w}\in\Delta_L}\sum_{\ell\in\mathscr{L}} w_\ell\, c^{(\ell)}(\mathbf{x}),
    \label{eq:dual_max}
\end{equation}
so $\mathscr{P}_{2}$ is $\min_{\mathbf{x}}\max_{\mathbf{w}\in\Delta_L}\mathbf{w}^{\top}\mathbf{c}(\mathbf{x})$. We approach the saddle point by alternating an inner step on $\mathbf{x}$ for fixed $\mathbf{w}$, which is the step handed to the quantum sampler, with an exponentiated-gradient step on $\mathbf{w}$~\cite{Arora2012}:
\begin{align}
    \mathbf{x}^{(t)} &\in \arg\min_{\mathbf{x}\ \text{s.t. D1--D4}}\ \sum_{\ell} w^{(t)}_\ell\, c^{(\ell)}(\mathbf{x}),
    \label{eq:inner}\\
    \eta^{(t)} &= \max_{\ell} c^{(\ell)}(\mathbf{x}^{(t)}),
    \label{eq:cert}\\
    w^{(t+1)}_\ell &\propto w^{(t)}_\ell \exp\!\bigl(\mu\, c^{(\ell)}(\mathbf{x}^{(t)})/\eta^{(t)}\bigr),
    \label{eq:eg}
\end{align}
with step $\mu>0$ and $\mathbf{w}^{(1)}$ uniform, for $T$ outer iterations; the weights move toward the scenarios that hurt the current schedule most. Equation~\eqref{eq:cert} is the exact worst-case AoII of $\mathbf{x}^{(t)}$ over $\mathscr{L}$, computed in $\mathcal{O}(L|\mathscr{U}|)$; SENTINEL returns the iterate with the smallest $\eta^{(t)}$ and that value as its certificate. Since the inner problem~\eqref{eq:inner} has a linear objective, it needs neither an epigraph variable nor scenario slacks.

%-----------------------------------------------------------------------
\subsection{Candidate Pruning}
\label{sec:pruning}
%-----------------------------------------------------------------------

The pre-elimination of Section~\ref{sec:reform} removes inadmissible combinations. We further retain, for each UE, only the $\varkappa$ admissible combinations with the smallest worst-case cost $\max_\ell\psi^{(\ell)}_{imqh}$, plus the best combination at every other visible node so that the capacity coupling is preserved, and denote the result $\mathscr{C}_i\subseteq\mathscr{A}_i$. The number of decision qubits per UE is then at most $\varkappa+|\mathscr{M}|$ instead of $\sum_m Q_mH_i$. Pruning could in principle discard an optimal combination that is individually poor but relieves a congested node; with $\varkappa=4$ this did not happen on any instance of Section~\ref{sec:performance}, where the exact optima over $\mathscr{C}$ and $\mathscr{A}$ coincided from 10 to 100 UEs.

%-----------------------------------------------------------------------
\subsection{Compact Domain QUBO}
\label{sec:compact_qubo}
%-----------------------------------------------------------------------

For domain $k$ with weights $\mathbf{w}$ and ADMM state $(\mathbf{z}_k,\boldsymbol{\upsilon}_k)$, let $\bar\psi_v=\frac{1}{|\mathscr{U}|}\sum_\ell w_\ell\psi^{(\ell)}_v$ for $v=(i,m,q,h)$, $i\in\mathscr{U}_k$, $(m,q,h)\in\mathscr{C}_i$, and let $b_m(\mathbf{x}_k)=\sum_{i\in\mathscr{U}_k}\sum_{(m,q,h)\in\mathscr{C}_i} \bar\beta^q_m x_{imqh}$ be the bandwidth that domain $k$ uses at node $m$. Satellites visible only to domain $k$ form $\mathscr{S}^{\mathrm{o}}_k$; those shared with other domains form $\mathscr{S}^{\mathrm{sh}}_k$. The binary vector $\mathbf{y}_k=(\mathbf{x}_k,\boldsymbol{\nu}_k)$ holds the decision variables and, for each node $m\in\mathscr{S}^{\mathrm{o}}_k\cup\mathscr{B}_k$, a capacity slack $\nu_m=\Delta_\beta\sum_{e=0}^{W_m-1}2^{e} y_{m,e}$ of $W_m=\lceil\log_2(B_m/\Delta_\beta+1)\rceil$ bits, where the bandwidth levels and capacities are integer multiples of $\Delta_\beta$. The domain QUBO is
\begin{align}
    F_k(\mathbf{y}_k)=\sum_{v}\bar\psi_v x_v
    &+\lambda\!\!\sum_{m\in\mathscr{S}^{\mathrm{o}}_k\cup\mathscr{B}_k}\!\!\bigl(b_m(\mathbf{x}_k)+\nu_m-B_m\bigr)^2
    \nonumber\\
    &+\frac{\lambda_{\mathrm{c}}}{2}\sum_{j\in\mathscr{S}^{\mathrm{sh}}_k}\bigl(b_j(\mathbf{x}_k)-z_{jk}+\upsilon_{jk}\bigr)^2 .
    \label{eq:Fk_compact}
\end{align}
The first term is the weighted AoII cost, the second enforces the capacity of the nodes that belong to the domain alone, and the third is the augmented-Lagrangian coordination term for the shared satellites, with penalty parameter $\lambda_{\mathrm{c}}$. Constraint D1 does not appear: the mixer of Section~\ref{sec:qaoa} keeps every UE group one-hot, so on the subspace the circuit can reach, $x_vx_{v'}=0$ for two candidates of the same UE and the corresponding entries of the QUBO matrix are dropped. The matrix $\mathbf{Q}^{(k)}$ collects the linear coefficients of $F_k$ on its diagonal, since $y^2=y$ for binary $y$, and the pairwise coefficients split between the two symmetric off-diagonal entries.

\begin{proposition}
\label{prop:penalty}
Let $\bar a_k=\sum_{i\in\mathscr{U}_k}\max_{v\in\mathscr{C}_i}\bar\psi_v$ and let $\Gamma_k^{\max}$ and $\Gamma_k^{\min}$ be the largest and smallest values of the coordination term over one-hot assignments of domain $k$. If
\begin{equation}
    \lambda>\frac{\bar a_k+\Gamma_k^{\max}-\Gamma_k^{\min}}{\Delta_\beta^{2}}
    \label{eq:lambda_sufficient}
\end{equation}
and the domain problem is feasible, then every one-hot minimizer of $F_k$ satisfies the capacity constraints of $\mathscr{S}^{\mathrm{o}}_k\cup\mathscr{B}_k$ and minimizes the weighted cost plus coordination term over the feasible one-hot assignments.
\end{proposition}

\begin{proof}
Over one-hot assignments the cost term lies in $[0,\bar a_k]$. For a feasible assignment whose slack bits represent the unused capacity exactly, the penalty vanishes and $F_k\leq\bar a_k+\Gamma_k^{\max}$. For one that exceeds the capacity of some node, the residual is a positive integer multiple of $\Delta_\beta$ for every non-negative slack, so $F_k\geq\lambda\Delta_\beta^2+\Gamma_k^{\min}>\bar a_k+\Gamma_k^{\max}$. No infeasible assignment is a minimizer, and on the feasible set $F_k$ equals the cost plus coordination term once the slack bits cancel the penalty.
\end{proof}

The bound is taken over one-hot assignments because that is the subspace the circuit explores. Taking it over the whole hypercube, as a penalty analysis of a transverse-field QAOA would require, inflates $\lambda$ by more than an order of magnitude and flattens the energy landscape to the point where simulated annealing finds no feasible point; this is the behavior of the epigraph encoding and is quantified in Section~\ref{sec:exact_validation}. The number of qubits of domain $k$ is
\begin{equation}
    \hat{N}_k=\sum_{i\in\mathscr{U}_k}|\mathscr{C}_i|+\sum_{m\in\mathscr{S}^{\mathrm{o}}_k\cup\mathscr{B}_k}W_m,
    \label{eq:Nk_compact}
\end{equation}
independent of $L$. With the parameters of Section~\ref{sec:performance} and two domains, the largest domain of the network-level instances of Section~\ref{sec:ue_sweep} has a median of 20, 33, 46 and 58 qubits at 5, 10, 15 and 20 UEs (12 to 76 over all seeds), against 120 to 260 for the epigraph encoding; the validation instances of Section~\ref{sec:exact_validation}, which draw other topologies and candidate sets, have 44 and 78 qubits at 10 and 25 UEs against 213 and 353. The ratio is four to six at 10 to 100 UEs and larger below 10 UEs, where the epigraph's fixed registers dominate.

\emph{Greedy assignment.} The classical reference used throughout, as warm start, as repair, and as the baseline of Section~\ref{sec:exact_validation}, processes the UEs of a domain in decreasing order of the spread of their candidate costs and gives each its cheapest candidate $v$ by $\bar\psi_v$ whose node still has capacity, falling back to the cheapest candidate when none has.

%-----------------------------------------------------------------------
\subsection{Ising Mapping and XY-Mixer QAOA}
\label{sec:qaoa}
%-----------------------------------------------------------------------

The QUBO is mapped to an Ising Hamiltonian through the substitution $y_a=(1-\sigma^{z}_{a})/2$, where $a$ indexes the $\hat N_k$ qubits of the domain and $\sigma^z_a$ is the Pauli-$Z$ operator on qubit $a$. Discarding additive constants,
\begin{align}
    H_{C}^{(k)}&=\sum_{a} h_{a}\,\sigma^{z}_{a}+\sum_{a<a'}J_{aa'}\,\sigma^{z}_{a}\sigma^{z}_{a'},
    \nonumber\\
    h_{a}&=-\frac{1}{2}\sum_{a'}Q^{(k)}_{aa'},\qquad J_{aa'}=\frac{1}{2}Q^{(k)}_{aa'},
    \label{eq:Hc}
\end{align}
where the sums over $a'$ run over the symmetric matrix $\mathbf{Q}^{(k)}$. Couplings $J_{aa'}$ arise only between candidates using the same node, and between those and the node's slack bits, so the cost layer is sparse: the candidates of one UE are not coupled to each other, nor are candidates at different nodes.

For the mixer we do not use the transverse field $\sum_a\sigma^x_a$, which would leave the one-hot subspace and force D1 back into the cost as a penalty. Instead, on each UE group we apply a ring of XY interactions and keep the transverse field only on the slack bits~\cite{Wang2020}:
\begin{equation}
    H_{M}^{(k)}=\sum_{i\in\mathscr{U}_k}\ \sum_{(a,a')\in\mathrm{ring}(\mathscr{C}_i)}
    \tfrac{1}{2}\bigl(\sigma^x_a\sigma^x_{a'}+\sigma^y_a\sigma^y_{a'}\bigr)
    +\sum_{a\in\mathrm{slack}}\sigma^x_a ,
    \label{eq:HM_xy}
\end{equation}
where $\mathrm{ring}(\cdot)$ denotes adjacent pairs in a cyclic ordering of the group. Each XY term conserves the Hamming weight of its group, so every Trotter factor of $e^{-\mathrm{i}\theta H_M^{(k)}}$ does as well, and a state initialized in the one-hot subspace stays in it exactly, so a measured sample can violate D1 only through hardware noise. Each group is initialized either in its $W$ state, the uniform superposition of its one-hot basis states, prepared with a linear-depth ladder of controlled rotations, or in the basis state of the greedy assignment of Section~\ref{sec:compact_qubo}, which we call the warm start. The depth-$p$ state is
\begin{equation}
    |\boldsymbol{\gamma},\boldsymbol{\theta}\rangle
    =\prod_{l=1}^{p}e^{-\mathrm{i}\theta_{l}H_{M}^{(k)}}e^{-\mathrm{i}\gamma_{l}H_{C}^{(k)}}\,|\psi_0\rangle,
    \label{eq:qaoa_state}
\end{equation}
and the angles $(\boldsymbol{\gamma},\boldsymbol{\theta})$ are optimized by a derivative-free classical optimizer (constrained optimization by linear approximation, COBYLA, in our evaluation) on the sampled energy $\langle H_C^{(k)}\rangle$. The sampled bitstrings are decoded classically; samples that are not one-hot or exceed a capacity are discarded, and the feasible sample with the smallest value of~\eqref{eq:Fk_compact} is returned as $\hat{\mathbf{x}}_k$. When $\hat N_k$ exceeds the capacity of the available device or simulator, the inner step uses a one-hot-preserving simulated annealer on the same QUBO.

%-----------------------------------------------------------------------
\subsection{Distributed Coordination}
\label{sec:admm}
%-----------------------------------------------------------------------

The domains are coupled through the satellites they share. For a shared satellite $s_j$, let $\mathscr{K}_j\subseteq\mathscr{K}$ be the $K_j$ domains that see it. The coordinator assigns each $k\in\mathscr{K}_j$ a bandwidth target $z_{jk}$ and keeps a scaled dual variable $\upsilon_{jk}$; the local QUBOs carry the coordination term of~\eqref{eq:Fk_compact}. After the domains return $\hat{\mathbf{x}}_k^{(r+1)}$ in round $r$, the coordinator sets $\tilde{a}_{jk}=b_j(\hat{\mathbf{x}}_{k}^{(r+1)})+\upsilon_{jk}^{(r)}$ and applies the sharing-form ADMM updates~\cite{Boyd2011}
\begin{align}
    z_{jk}^{(r+1)}&=\tilde{a}_{jk}-\frac{1}{K_j}\max\Bigl\{0,\;\sum_{k'\in\mathscr{K}_j}\tilde{a}_{jk'}-B^{s}_{j}\Bigr\},
    \label{eq:z_update}\\
    \upsilon_{jk}^{(r+1)}&=\upsilon_{jk}^{(r)}+b_j\bigl(\hat{\mathbf{x}}_{k}^{(r+1)}\bigr)-z_{jk}^{(r+1)},
    \label{eq:u_update}
\end{align}
which lower every target at an overloaded satellite by an equal share of the excess and accumulate the remainder in the duals. The loop ends when both residuals fall below a tolerance or after $R_{\max}$ rounds; $\lambda_{\mathrm{c}}$ is adapted by residual balancing.

Two details decide whether this loop produces good schedules. The targets are initialized from an uncoupled round ($\lambda_{\mathrm{c}}=0$) projected onto capacity, not from the equal share $B^{s}_j/K_j$, which pulls every domain toward using exactly its share of every shared satellite whether it needs it or not and left schedules 60 to 560\% above the optimum at 50 to 100 UEs. And $\lambda_{\mathrm{c}}$ is scaled to the mean per-UE spread of the weighted costs $\bar\psi_v$, so that one bandwidth unit of deviation costs about one typical per-UE cost difference; a generic penalty two orders of magnitude larger lets the coordination term dominate everything else the inner solver sees. If no capacity-feasible network schedule exists when the loop stops, a greedy repair moves UEs from overloaded satellites to the candidate with the smallest increase in worst-case cost at a node with spare capacity (Algorithm~\ref{alg:SENTINEL}).

\begin{algorithm}[!t]
\caption{SENTINEL: Distributed Robust Orchestration for VANGUARD}
\label{alg:SENTINEL}
\footnotesize
\SetKwInOut{Input}{Input}
\SetKwInOut{Output}{Output}
\Input{Network, sources, scenario set $\mathscr{L}$, $\varkappa$, $T$, $\mu$, $p$, $R_{\max}$, $\epsilon_{\mathrm{tol}}$.}
\Output{$\mathbf{x}^{\star}$ decoded into $\rho_i$, $\alpha_{ij}$, $\omega_{in}$, $\beta^{us}_{ij}$,
$\beta^{ub}_{in}$; certificate $\eta^{\star}$; worst scenario $\ell^{\star}$.}
\nl Partition the network into $K$ domains; tabulate $\Psi_i(\cdot)$ by~\eqref{eq:psi_def}, the costs
    $\psi^{(\ell)}_{imqh}$ by~\eqref{eq:psi_imqh}, and the candidates $\mathscr{C}_i$.\\
\nl $\mathbf{w}^{(1)}\leftarrow$ uniform; $\eta^{\star}\leftarrow+\infty$.\\
\nl \For{$t=1,\ldots,T$}{
\nl Warm start: solve each domain with $\lambda_{\mathrm{c}}=0$; set $z_{jk}^{(0)}$ to the projected loads, $\upsilon_{jk}^{(0)}\leftarrow 0$;
    scale $\lambda_{\mathrm{c}}$ to the cost spread.\\
\nl \For{$r=0,\ldots,R_{\max}-1$}{
\nl \For{each domain $k$ \emph{in parallel}}{
\nl Build $F_k$ by~\eqref{eq:Fk_compact} with $\mathbf{w}^{(t)}$, $\mathbf{z}_k^{(r)}$, $\boldsymbol{\upsilon}_k^{(r)}$; map to
    $H_C^{(k)}$ by~\eqref{eq:Hc}.\\
\nl \lIf{$\hat N_k$ fits the sampler}{run XY-mixer QAOA~\eqref{eq:qaoa_state}, decode, keep the best feasible
    sample as $\hat{\mathbf{x}}_k^{(r+1)}$}
\nl \lElse{obtain $\hat{\mathbf{x}}_k^{(r+1)}$ by one-hot-preserving annealing of $F_k$}
}
\nl Update $z_{jk}^{(r+1)}$, $\upsilon_{jk}^{(r+1)}$ by~\eqref{eq:z_update}--\eqref{eq:u_update}; adapt $\lambda_{\mathrm{c}}$.\\
\nl \If{$\hat{\mathbf{x}}^{(r+1)}$ satisfies D1--D3 and $\max_\ell c^{(\ell)}(\hat{\mathbf{x}}^{(r+1)})<\eta^{\star}$}{
    $\mathbf{x}^{\star}\leftarrow\hat{\mathbf{x}}^{(r+1)}$; $\eta^{\star}\leftarrow\max_\ell c^{(\ell)}(\mathbf{x}^{\star})$.}
\nl \lIf{residuals $\leq\epsilon_{\mathrm{tol}}$}{\textbf{break}}
}
\nl \lIf{no feasible schedule in this iteration}{apply greedy repair, evaluate~\eqref{eq:cert}, and update $\mathbf{x}^{\star}$, $\eta^{\star}$ if improved}
\nl $\mathbf{x}^{(t)}\leftarrow$ best schedule of this iteration; update $\mathbf{w}^{(t+1)}$ by~\eqref{eq:eg}.\\
}
\nl Decode $\mathbf{x}^{\star}$; $\ell^{\star}\leftarrow\arg\max_{\ell}c^{(\ell)}(\mathbf{x}^{\star})$.\\
\end{algorithm}

\textbf{Convergence and certification.} For convex local problems the sharing form of ADMM converges~\cite{Boyd2011}; here the local problems are binary, so the loop is a heuristic and neither it nor the outer weight iteration is guaranteed to reach the saddle point of~\eqref{eq:dual_max}. SENTINEL does not rely on either for correctness. The incumbent is replaced only by network schedules satisfying D1 to D3, and its worst-case AoII~\eqref{eq:cert} is computed exactly, so the iteration count affects the quality of the schedule but never the validity of its certificate.

%-----------------------------------------------------------------------
\subsection{Resource Scaling and Complexity}
\label{sec:complexity}
%-----------------------------------------------------------------------

\textbf{Qubits and gates.} By~\eqref{eq:Nk_compact} the qubit count of a domain grows with $\varkappa|\mathscr{U}_k|$ and the slack registers and does not depend on $L$. Because the cost Hamiltonian couples only candidates that share a node, the cost layer of domain $k$ needs at most $\sum_{m}\binom{n_{m}+W_m}{2}$ two-qubit $ZZ$ rotations, with $n_m$ the number of candidates at node $m$ in the domain and $W_m=0$ for shared satellites; the XY mixer adds $|\mathscr{C}_i|$ two-qubit rotations per UE group and layer and the $W$-state preparation $|\mathscr{C}_i|-1$. Routed on a heavy-hex processor at depth 1, domains of 12 to 24 qubits required 58 to 305 two-qubit gates and a 38-qubit domain 606 (Section~\ref{sec:device}); the epigraph encoding, whose scenario penalties couple every pair of local variables, needs $\hat N_k(\hat N_k-1)/2$ per layer with $\hat N_k$ four to six times larger.

\textbf{Time.} Tabulating $\Psi_i$ costs $\mathcal{O}(Y^{\max}|\mathcal{X}_i|^3+|\mathcal{X}_i|(Y^{\max})^2)$ per UE and the coefficients $\mathcal{O}(L\sum_i|\mathscr{A}_i|)$ lookups, once per set of sources and scenarios; assembling $F_k$ costs $\mathcal{O}(\hat N_k^2)$, the ADMM updates $\mathcal{O}(\sum_j K_j)$ per round, and the certificate $\mathcal{O}(L|\mathscr{U}|)$. With $I_{\mathrm{QAOA}}$ optimizer iterations of $N_{\mathrm{shot}}$ shots and the domains in parallel, a solve costs
\begin{equation}
    C_{\mathrm{SENTINEL}}=\mathcal{O}\!\Bigl(T\,R_{\max}\bigl(I_{\mathrm{QAOA}}\,p\,N_{\mathrm{shot}}\max_k G_k
    +L|\mathscr{U}|\bigr)\Bigr),
    \label{eq:complexity}
\end{equation}
where $G_k$ is the gate count above. In the evaluation a 20-UE solve with $T=6$ and $R_{\max}=20$ took 1 to 7~s on the classical annealer, and the orchestration loop used 34 to 223~s of solver time per ten-minute run. This time is reported but not charged to the AoII ($\epsilon_i=0$, Section~\ref{sec:orchestration}); a reactive re-solve after a link failure likewise takes effect at the next orchestration interval, which is optimistic by the duration of that solve. Exhaustive search over $\mathscr{P}_2$ costs $\mathcal{O}(L|\mathscr{U}|\prod_i|\mathscr{A}_i|)$, branch-and-cut on $\mathscr{P}_3$ has the same worst case, although it is fast on our instances, and a centralized QAOA would need $\sum_k\hat N_k$ qubits on one device; the advantage of the decomposition is the qubit count per device, not asymptotic time.
%=======================================================================
\section{Performance Evaluation}
\label{sec:performance}
%=======================================================================

We evaluate SENTINEL with the simulator released with this paper\footnote{\url{https://github.com/tanzimulalam/Quantum_SAGIN_AoII_Updated_Simulation}, tag \texttt{v2.3.5-manuscript}; the script \texttt{scripts/reproduce\_all.sh} regenerates every table and figure.} under the parameters of Table~\ref{tab:simulation_parameters}, which hold for every figure. Each configuration runs for 60{,}000 slots of 10~ms (ten minutes of network time) on 30 independent seeds; the policies of one seed share its topology, source trajectories, fading, and disruption realization, so all comparisons are paired. We report medians with interquartile ranges (IQRs), because a few seeds contain long NLOS passes that dominate means; means with 95\% confidence intervals, per-seed results, and paired test tables are in the repository.

\begin{table}[t]
\centering
\caption{Simulation Parameters}
\label{tab:simulation_parameters}
\scriptsize
\renewcommand{\arraystretch}{1.0}
\begin{tabular}{p{0.43\columnwidth} p{0.47\columnwidth}}
\hline
\textbf{Parameter} & \textbf{Value} \\
\hline
\multicolumn{2}{l}{\emph{Network}} \\
Number of UEs & $\{5,10,15,20\}$; up to 100 in Sec.~\ref{sec:exact_validation} \\
BSs / satellite serving slots & 2 / 2 (topology study: 0 to 3 each) \\
Area; BS height; UE height & $20\times20$~km$^2$; 25~m; 1.5~m \\
Satellite passes & Starlink TLE snapshot, SGP4, MTSU site (35.85$^\circ$N, 86.39$^\circ$W); min.\ elevation $25^\circ$ \\
\hline
\multicolumn{2}{l}{\emph{Channel and resources}} \\
UE transmit power & 23~dBm \\
UE--BS link & 3.5~GHz, 20~MHz, $-147.55+20\log_{10}f+26\log_{10}d+6$~dB, gain 15~dBi, NF 7~dB, SE cap 6 \\
UE--SAT link & 12~GHz, 50~MHz, free space $+1$~dB, gains 30 $+$ 30~dBi, NF 9~dB, SE cap 4 \\
NTN large-scale model & TR~38.811 suburban (LOS prob., clutter, shadow fading per pass); atm.\ $0.5$~dB$/\sin V$ \\
Bandwidth levels $\bar\beta^q_m$ ($Q_m=3$) & $\{2.5,5,10\}$~MHz, $\Delta_\beta=2.5$~MHz \\
\hline
\multicolumn{2}{l}{\emph{Source and AoII}} \\
Source states; transition matrix & 4; birth--death, self-transition $\mathcal{U}[0.90,0.99]$ per UE \\
Distortion; AoII budget $\Delta^{II,\max}_i$ & $|x-\hat x|$; 10 \\
Update periods ($H_i=4$); slot; packet & $\{5,10,20,50\}$ slots; 10~ms; 20~kbit \\
\hline
\multicolumn{2}{l}{\emph{Disruption}} \\
Scenario set ($L=5$) & nominal; deep shadowing 12~dB for 30 slots; handover of each SAT (20 slots); joint loss (30 slots) \\
Realized process (per 1~s interval) & TR~38.811 state per pass; shadowing event per link w.p.\ 0.07; handover failure per SAT w.p.\ 0.03 \\
Link-failure handling (all policies) & node blocked after 60 failed slots, for 6~s or until its next handover \\
\hline
\multicolumn{2}{l}{\emph{Solver}} \\
Formulation & compact dual, $\varkappa=4$, $T=6$ outer iterations, $\mu=2$ \\
Local solver & one-hot annealing (all network-level runs); XY-QAOA, depth 1, 512 shots, COBYLA 40 it.\ on domains of $\le 26$ qubits (Sec.~\ref{sec:exact_validation}) \\
ADMM $\lambda_{\mathrm{c},0}$, tolerance, rounds & scaled to cost spread, $10^{-2}$, 20 \\
Quantum backend & Qiskit Aer~\cite{Qiskit2024}; FakeTorino / FakeFez; \texttt{ibm\_kingston}, \texttt{ibm\_fez} \\
\hline
\end{tabular}
\end{table}

\textbf{Baselines.} \emph{Random} transmits with probability 0.1 per slot over a uniformly chosen visible node with equal bandwidth sharing. \emph{AoII threshold} is a state-aware policy that transmits whenever its AoII is positive, at most once per 5 slots, on the node with the best nominal rate. \emph{Whittle index} is the restless-bandit scheduler of the AoII literature~\cite{Whittle1988, Kriouile2021}: the index of each UE's (source state, estimate, age) arm is computed per source by relative value iteration on a grid of transmission prices, and every slot the mismatched UEs are ranked by index and admitted at their best node while the node's shared bandwidth still delivers every admitted packet; it has the threshold policy's observations and rate cap. \emph{PPO} is an actor-critic agent trained by proximal policy optimization (PPO)~\cite{Schulman2017} that selects, once per orchestration interval, a serving node and an update-period level for every UE, with invisible or failed nodes masked and equal bandwidth sharing, rewarded with the negative mean AoII of the interval plus a small delivery bonus; it is trained for 5000 episodes per network size on 30 seeds disjoint from the evaluation seeds and evaluated deterministically on the same episodes as every other policy. Like SENTINEL, it acts once per interval, so the comparison isolates the decision rule; the per-slot, source-observing role is filled by the threshold and Whittle policies. \emph{Nominal SENTINEL} runs Algorithm~\ref{alg:SENTINEL} with the nominal scenario alone ($L=1$); \emph{SENTINEL} is the robust schedule with the five-scenario set; \emph{SENTINEL + trigger} uses SENTINEL's association and bandwidth with transmissions triggered by the observed mismatch, rate-capped at the chosen period. Every policy receives the same acknowledgement feedback and link-failure rule of Section~\ref{sec:feedback}. The exact reference for Section~\ref{sec:exact_validation} is the MILP $\mathscr{P}_3$ solved with HiGHS~\cite{HiGHS2018}.

%-----------------------------------------------------------------------
\subsection{AoII Versus Network Size}
\label{sec:ue_sweep}
%-----------------------------------------------------------------------

\begin{figure}[!t]
\centering
\includegraphics[width=\columnwidth]{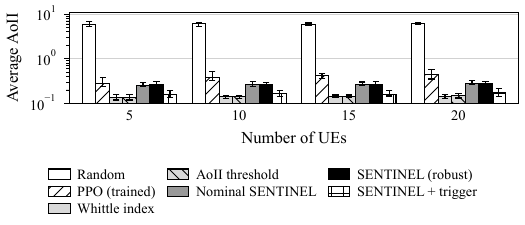}
\caption{Median average AoII versus the number of UEs (whiskers: IQR over 30 seeds).}
\label{fig:ue_sweep}
\end{figure}

Fig.~\ref{fig:ue_sweep} and Table~\ref{tab:tle} give the median average AoII from 5 to 20 UEs. In every network-level result of this section the inner solver of Algorithm~\ref{alg:SENTINEL} was the one-hot-preserving annealer, because the domains span 12 to 76 qubits and are re-solved hundreds of times per run; the XY-mixer QAOA is validated against it and the exact optimum on domains of up to 24 qubits (Section~\ref{sec:exact_validation}) and run on hardware (Section~\ref{sec:device}). The state-aware per-slot policies, AoII threshold and Whittle index, are lowest at 0.13 to 0.15 and coincide (Section~\ref{sec:topology}). The SENTINEL schedules, robust and nominal, sit at 0.25 to 0.28 and do not grow with the number of UEs, because bandwidth levels are dedicated and a re-solve redistributes load as satellites come and go. SENTINEL with the trigger lies between the two at 0.155 to 0.175: it keeps the robust association and bandwidth but transmits only when the estimate is wrong, raising the informative-update ratio from 0.23 to 0.96. Random scheduling sits at 6; the PPO agent is 11\% worse than the nominal schedule at 5 UEs (0.280 against 0.253), degrades to 0.446 at 20 UEs as its action space grows, and routes all traffic to satellites. Delivery ratios are 0.92 to 0.94 for every policy except random; satellites carry 73 to 79\% of robust SENTINEL's deliveries and 97 to 98\% of the threshold policy's.

\begin{table}[!t]
\centering
\caption{Median Average AoII Versus Number of UEs (30 Seeds)}
\label{tab:tle}
\footnotesize
\setlength{\tabcolsep}{3pt}
\resizebox{\columnwidth}{!}{%
\begin{tabular}{|c|c|c|c|c|c|c|c|}
\hline
\textbf{UEs} & \textbf{Random} & \textbf{PPO} & \textbf{Whittle} & \textbf{AoII thr.} & \textbf{Nominal} & \textbf{SENTINEL} & \textbf{+ trigger} \\
\hline
5  & 5.94 & 0.280 & 0.134 & 0.134 & 0.253 & 0.265 & 0.155 \\
10 & 6.06 & 0.372 & 0.141 & 0.142 & 0.270 & 0.269 & 0.168 \\
15 & 5.87 & 0.413 & 0.142 & 0.142 & 0.269 & 0.264 & 0.155 \\
20 & 6.05 & 0.446 & 0.146 & 0.149 & 0.276 & 0.278 & 0.175 \\
\hline
\end{tabular}%
}
\end{table}

%-----------------------------------------------------------------------
\subsection{The Certificate and the Value of Robustness}
\label{sec:certificate}
%-----------------------------------------------------------------------

\begin{figure}[!t]
\centering
\includegraphics[width=0.66\columnwidth]{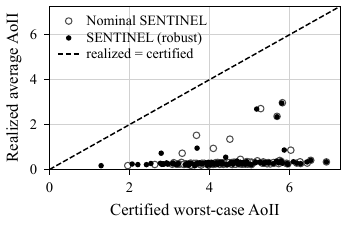}
\caption{Certified worst-case AoII versus realized average AoII per orchestration run, robust and nominal SENTINEL.}
\label{fig:certificate}
\end{figure}

Fig.~\ref{fig:certificate} plots, for every run, the time-weighted certified worst case against the realized average AoII; the nominal schedule, although planned for the nominal scenario alone, is certified against the same five-scenario set as the robust one. The realized average stayed within the certificate in 98.7\% of orchestration epochs for the robust schedule and 98.5\% for the nominal one; the violations are NLOS passes, whose 17 to 25~dB clutter loss exceeds the 12~dB deep-shadowing scenario. The certificate is conservative: the certified worst case exceeds the realized average by a median factor of 17 per epoch (IQR 13 to 23 over 3975 epochs; 16 to 19 by network size), because every scenario applies its outage to every inter-delivery interval of the epoch, a persistent disruption, whereas realized disruptions are intermittent. The 1.3\% of violating epochs are not marginal: a blocked pass pushes realized AoII to 1.2 to 2.1 times the certificate at the median and up to 5.4 times, so the certificate is a guarantee against the scenario set, not a prediction of typical behavior. It is proved for the periodic schedule, whose deliveries do not depend on the source state; the triggered variant transmits on observed mismatch, which the renewal argument of Section~\ref{sec:renewal} does not cover, so we check empirically: its realized average was within its schedule's certificate in 97.9 to 98.7\% of epochs.

The benefit of the scenario set appears in the certificate, not in the realized AoII. Across the 30 paired seeds the robust schedule's certified worst case is 2 to 7\% below the nominal schedule's (paired-median ratio 0.93 at 5 UEs to 0.98 at 20 UEs), lower in 67 to 100\% of seeds, with Wilcoxon signed-rank $p\le 9\times10^{-5}$ at every size after Holm correction. The realized AoII of the two is statistically indistinguishable at every size ($p\ge0.08$), although the robust mean at 20 UEs is lower (0.39 against 0.46) because the robust schedule keeps 21 to 27\% of deliveries on base stations. The gap between SENTINEL with the trigger and the threshold policy, a Hodges--Lehmann paired difference of 0.02 to 0.03 in AoII, is small but significant at every size ($p\le0.013$), and PPO is significantly behind SENTINEL from 10 UEs on.

\begin{figure}[!t]
\centering
\includegraphics[width=\columnwidth]{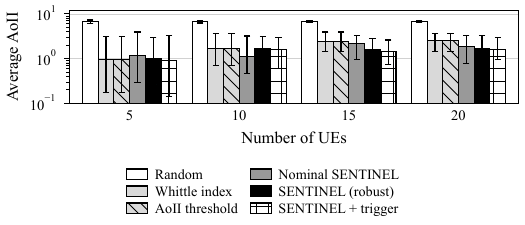}
\caption{Median average AoII without acknowledgement feedback or link-failure handling (60 seeds).}
\label{fig:noack}
\end{figure}

\textbf{Without failure feedback.} Whether the network can react to a failed link matters more than any scheduling choice. Fig.~\ref{fig:noack} repeats the sweep with the link-failure rule disabled for every policy, on 60 seeds because the distributions become heavy-tailed. Median AoII rises by an order of magnitude for all policies and the certificate holds in 87 to 92\% of epochs. The threshold and Whittle-index policies, which coincide here as well and concentrate 98\% of their traffic on the best nominal satellite, tie the SENTINEL schedules at 5 and 10 UEs and fall 25 to 30\% behind SENTINEL with the trigger at 15 and 20 UEs (Hodges--Lehmann differences $-0.77$ and $-0.71$, Holm-corrected $p\le0.043$ against both), while the SENTINEL schedules keep 18 to 27\% of deliveries on base stations. Robust and nominal SENTINEL remain indistinguishable in realized AoII ($p\ge0.41$). The measurable benefit of the robust scenario set is its lower certificate in both regimes, whereas a scheduled base-station share pays off without feedback.

\textbf{Scenario-set size.} With the list extended by a mild 6~dB shadowing scenario and truncated to $L=1,\ldots,6$ at 10 UEs, realized AoII is flat in $L$ (0.265 to 0.278), and the certificate and base-station hedge move only once the second shared satellite's handover enters the set (certified worst case over the full set 4.57 to 4.28, satellite share 0.79 to 0.74 between $L=4$ and $L=5$). The mild scenario never binds and no certificate changes on any seed or size, so we use five.

%-----------------------------------------------------------------------
\subsection{Topology, the Whittle Index, and Capacity}
\label{sec:topology}
%-----------------------------------------------------------------------

\begin{figure}[!t]
\centering
\includegraphics[width=\columnwidth]{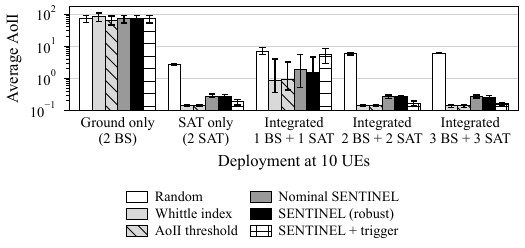}
\caption{Median average AoII at 10 UEs by deployment (log scale; whiskers: IQR over 30 seeds).}
\label{fig:topology}
\end{figure}

Fig.~\ref{fig:topology} compares deployments at 10 UEs. Two base stations cover only the UEs within 5 to 7~km of the $20\times20$~km$^2$ area, so ground-only monitoring fails for every policy (median AoII 63 to 81). Two satellite serving slots alone bring the per-slot policies to 0.14 and the SENTINEL schedules to 0.18 to 0.28, and the integrated $2+2$ and $3+3$ deployments are best for every policy, with SENTINEL routing 26 to 36\% of its deliveries to base stations. The one deployment in which the SENTINEL schedules are clearly worse than the per-slot policies is a single base station with a single satellite slot (medians 1.6 to 1.9 against 0.87 to 0.90, and 5.4 for SENTINEL with the trigger): with one slot every handover is a total outage for every UE on it, and a scheduled policy loses more during the recovery than one that retransmits the moment the link returns.

\begin{figure}[!t]
\centering
\includegraphics[width=0.72\columnwidth]{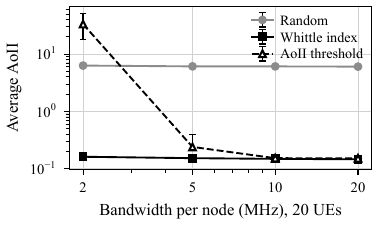}
\caption{Median average AoII of the Whittle-index and AoII-threshold policies at 20 UEs versus bandwidth per node.}
\label{fig:capacity}
\end{figure}

\textbf{Whittle index and capacity.} At the bandwidths of Table~\ref{tab:simulation_parameters} a node delivers 19 to 40 simultaneous equal sharers within one slot, so capacity never binds and the Whittle scheduler makes the same decisions as the threshold policy (identical AoII to three decimals at every size and topology, $p\ge0.18$). Fig.~\ref{fig:capacity} reduces the bandwidth per node at 20 UEs: at 5~MHz the threshold policy is 57\% worse than the index scheduler (0.239 against 0.152, $p=4\times10^{-9}$), and at 2~MHz its simultaneous transmissions collapse every shared link (median AoII 33.5, delivery ratio 0.16) while the capacity-aware scheduler stays at 0.16; SENTINEL's dedicated bandwidth levels avoid this regime by construction. The index scheduler's one weakness is the ground-only deployment (81 against 63), where it never tries a node that cannot deliver at the nominal rate.

%-----------------------------------------------------------------------
\subsection{Validation Against the Exact Minimax Optimum}
\label{sec:exact_validation}
%-----------------------------------------------------------------------

Table~\ref{tab:optimality_gap} reports the mean gap to the exact MILP optimum of $\mathscr{P}_3$ over capacity-feasible schedules (10 seeds; 5 for the QAOA rows), for the XY-mixer QAOA on the statevector simulator, for the full distributed pipeline of Algorithm~\ref{alg:SENTINEL} with the compact QUBO and with the epigraph encoding of Section~\ref{sec:epigraph}, and for greedy association refined by local search (LS).
The XY-mixer QAOA returned the exact optimum on every instance whose domains fit the simulator (8 to 24 qubits), with 100\% of samples one-hot and capacity-feasible; the epigraph encoding, run with the same classical search, is 38 to 51\% above the optimum on the same instances with 133 to 193 qubits, and a plain simulated annealer on its penalty landscape found no feasible assignment in 50 solves at 10 UEs. The compact distributed pipeline matches the exact optimum on all 10- and 25-UE instances and on the severity instances at 20 UEs, is within 3 to 7\% on the 20-UE scenario-count instances (optimal on 60 to 90\% of them) and within 12\% at 50 UEs, and is limited at 75 to 100 UEs by coordinating five shared satellites among six to eight domains, which neither more ADMM rounds nor smaller dual steps resolve; there, greedy association with local search is within 12 to 15\%. The exact reference reached proven optimality within a second up to 75 UEs and timed out on one instance in ten at 100 UEs.

\begin{table}[!t]
\centering
\caption{Mean Optimality Gap (\%) Above the Exact MILP Optimum of $\mathscr{P}_3$}
\label{tab:optimality_gap}
\scriptsize
\setlength{\tabcolsep}{2pt}
\begin{tabular}{|l|c|c|c|c|c|}
\hline
 & \textbf{Qubits/} & \textbf{XY-} & \textbf{Compact} & \textbf{Epigraph} & \textbf{Greedy} \\
\textbf{Setting} & \textbf{domain} & \textbf{QAOA} & \textbf{ADMM} & \textbf{ADMM} & \textbf{+LS} \\
\hline
4--8 UEs, 2+2 & 8--24 & 0.0 (15/15) & 0.0 & 38--51 & 0.0 \\
10 / 25 UEs, 2+2 & 44 / 78 & -- & 0.0 & 42.7 / 1.8 & 0.0 \\
50 UEs, 4+3 & 142 & -- & 11.5 & 38.8 & 6.0 \\
75 / 100 UEs & 159 / 205 & -- & 49.9 / 60.5 & -- & 14.4 / 11.5 \\
Severity, 20 UEs & 29--66 & -- & 0.0 & 0.2--4.7 & 0.0 \\
$L=1/3/5/6$, 20 UEs & 60--77 & -- & 0/3.2/6.9/0 & 74/233/9.0/3.8 & 0/0/1.1/0 \\
\hline
\end{tabular}
\end{table}

%-----------------------------------------------------------------------
\subsection{Quantum Resources and Device Results}
\label{sec:device}
%-----------------------------------------------------------------------

Table~\ref{tab:hardware} reports the sampling of the hardware-sized domains on the \texttt{ibm\_kingston} and \texttt{ibm\_fez} processors (Heron r2, 156 qubits; $\varkappa=4$, depth 1, 4096 shots per domain, transpiled at optimization level 3 with eight routing trials; 27 October to 1 November 2026). Every scenario weight, ADMM target, and QAOA angle was optimized on the simulator and each device job executed only the final sampling circuit of one or two domains; the 38-qubit domain exceeds the statevector simulator and received the angles of the largest simulated domain of its run, relying on parameter concentration~\cite{Brandao2018}. Per run the table lists the feasible-shot fraction, the probability of the optimum among the feasible shots, and the same probability under uniform sampling of the run's enumerated feasible one-hot set (each run re-solves the scenario weights, which changes the optimal set on the two seed-2 domains); the product of the first two is the probability that a single shot is optimal. The energy spread of the feasible set is listed because several domains are nearly flat: below 2\% any feasible assignment is within 2\% of the optimum, whereas on the other domains the spread is 35 to 307\% and uniform sampling reaches an approximation ratio of only 0.62 to 0.91.

\begin{table}[!t]
\centering
\caption{Device Sampling Results on IBM Heron Processors}
\label{tab:hardware}
\scriptsize
\setlength{\tabcolsep}{1.9pt}
\begin{tabular}{|l|c|cc|c|cc|cc|cc|}
\hline
 & & \multicolumn{2}{c|}{\textbf{2q gates}} & & \multicolumn{2}{c|}{\textbf{Feasible}} & \multicolumn{2}{c|}{$P(\text{opt.}|\text{feas.})$} & \multicolumn{2}{c|}{\textbf{Unif.} $P(\text{opt.})$} \\
\textbf{Instance} & \textbf{Qubits} & G & $W$ & \textbf{Spread} & G & $W$ & G & $W$ & G & $W$ \\
\hline
U6 s1 d0 & 18 & 139 & 179 & 43\% & 63\% & 62\% & 0.95 & 0.17 & 0.125 & 0.125 \\
U6 s1 d1 & 20 & 301 & 314 & 307\% & 43\% & 36\% & 0.93 & 0.15 & 0.125 & 0.125 \\
U6 s2 d0 & 18 & 139 & 179 & 35\% & 61\% & 62\% & 0.95 & 0.09 & 0.031 & 0.125 \\
U6 s2 d1 & 12 & 58 & 78 & 0\% & 74\% & 74\% & 0.94 & 0.12 & 0.016 & 0.125 \\
U6 s3 d0 & 24 & 305 & 333 & 115\% & 35\% & 34\% & 0.91 & 0.03 & 0.062 & 0.062 \\
U6 s3 d1 & 14 & 120 & 154 & 263\% & 70\% & 74\% & 0.96 & 0.20 & 0.250 & 0.250 \\
U6 s5 d0 & 12 & 58 & 82 & 0\% & 76\% & 73\% & 0.98 & 0.13 & 0.125 & 0.125 \\
U6 s5 d1 & 12 & 58 & 82 & 0\% & 76\% & 74\% & 0.98 & 0.12 & 0.125 & 0.125 \\
U6 s6 d0 & 8 & 38 & 52 & 1\% & 86\% & 65\% & 0.99 & 0.33 & 0.250 & 0.250 \\
U6 s6 d1 & 16 & 79 & 108 & 1\% & 62\% & 44\% & 0.95 & 0.13 & 0.062 & 0.062 \\
U6 s7 d0 & 12 & 58 & 80 & 2\% & 72\% & 59\% & 0.97 & 0.12 & 0.125 & 0.125 \\
U6 s7 d1 & 18 & 142 & 185 & 63\% & 59\% & 49\% & 0.93 & 0.07 & 0.125 & 0.125 \\
U10 s1 d0 & 38 & 606 & 719 & 82\% & 3\% & 3\% & 0.24 & 0.00 & 0.006 & 0.006 \\
\hline
\end{tabular}

\vspace{2pt}
\parbox{\columnwidth}{\scriptsize U: UEs, s: seed, d: domain; G: greedy start on \texttt{ibm\_kingston} (seeds 6, 7: \texttt{ibm\_fez}); $W$: $W$-state start, same processor. Uniform baseline computed per run.}
\end{table}

Three facts follow from the table. First, from the greedy warm start, which is itself optimal on every hardware-sized domain, the device returned an optimal assignment in 32 to 85\% of all shots on the 6-UE domains (91 to 99\% of the feasible shots, against 1.6 to 25\% for uniform feasible sampling). This measures the retention of an already-optimal warm start under noise, not a search. Second, at 38 qubits the per-shot probability of the optimum, 0.6\%, equals the uniform baseline: noise destroys 97\% of the shots, and the warm start shows only among the surviving feasible shots (24\% optimal against 0.6\%). Third, from the $W$ state, with no warm start, the device behaves as a uniform sampler of the feasible set: the probability of the optimum among feasible shots is 0.03 to 0.33 against the uniform 0.06 to 0.25 of the same run, between 0.52 and 2.04 times the uniform value on every 6-UE domain, and the approximation ratio of feasible shots exceeds the uniform value by at most 0.07; at 38 qubits no feasible shot was optimal. The 1.3-fold concentration of one noiseless layer is lost to hardware noise. Every decoded schedule still attained the MILP optimum of its instance (at 38 qubits the best domain sample was suboptimal but tied at the network worst case), because the feasible sets are small and the decoder keeps the best of thousands of shots.

The feasible-shot fraction tracks the noise: 59 to 86\% at 8 to 18 qubits with up to 142 two-qubit gates, 35 to 43\% at 20 to 24 qubits with about 300 gates, and 2.6\% at 38 qubits with 606 gates, bracketed by the FakeTorino and FakeFez noise models; the $W$-state circuits, 13 to 113 gates deeper, give fractions within 7 points of the greedy-start circuits on \texttt{ibm\_kingston} and 10 to 21 points lower on \texttt{ibm\_fez}. A few hundred more two-qubit gates would leave fewer than one feasible shot in a thousand, which places the present limit of the quantum step for these routed circuits at about 40 qubits per domain (one domain of a ten-UE network) at depth 1; hardware-tailored one-hot instances have reached 144 qubits with classical repair of infeasible samples~\cite{Bucher2026}. Error suppression does not raise it: re-sampling the 24- and 38-qubit domains in one calibration session with dynamical decoupling (XY4)~\cite{Viola1999}, and with decoupling plus Pauli twirling~\cite{Wallman2016}, lowered the feasible fraction from the in-session control values of 25.0\% and 2.4\% to 21.1\% and 0.7\%, and to 13.8\% and 0.2\%, because these circuits have almost no idle time in which decoupling could help and twirling inserts further gates. Because the unmitigated yield of one circuit varied by ten points between calibration sessions (34.8\% against 25.0\%), all device comparisons are made within one session.

We therefore claim no quantum advantage: in simulation the QAOA layers only concentrate a constrained sampler (from the $W$ state the probability of sampling an optimal assignment is 0.28 at depth 1 and 0.42 at depth 2, against 0.22 under uniform sampling of the one-hot subspace), and on every hardware-sized domain greedy association is already optimal. The instances on which it is not, 50 UEs and above in Table~\ref{tab:optimality_gap}, need 140 to 200 qubits per domain, beyond simulation and beyond the gate budget that the feasible-shot fractions above allow on current devices; whether a sampler at that size concentrates on the optimal set is left open.

%=======================================================================
\section{Conclusion}
\label{sec:conclusion}
%=======================================================================

This paper asked how a satellite-ground edge network should schedule status updates when what matters is whether the edge estimate is correct and the satellite links fail at each handover and for whole passes under blockage. We formulated the robust problem VANGUARD, reduced the history-dependent AoII to a tabulated cost per inter-delivery interval, and solved the resulting minimax assignment with SENTINEL, which dualizes the minimax into classically updated scenario weights, encodes each domain's linear subproblem as a compact QUBO whose one-hot structure the XY mixer preserves, coordinates the domains by ADMM, and certifies every schedule by its exact worst case. Under Starlink passes and the TR 38.811 channel the certificate held in 98.7\% of epochs, the robust certificate was 2 to 7\% below the nominal one at no cost in realized AoII, and SENTINEL with a state-aware trigger came within 10 to 20\% of the best per-slot heuristic and outperformed it by 25 to 30\% without failure feedback. The compact formulation needs four to six times fewer qubits than the epigraph encoding and matches the exact optimum on the 10- and 25-UE validation instances; thirteen device domains of 8 to 38 qubits decoded to the exact optimum, but at these sizes the quantum sampler is correct rather than a better optimizer than greedy association. Future work includes a coordinator that negotiates per-UE assignments to remove the bottleneck above 50 UEs, a scenario set learned from the site's pass statistics, and testing whether a constrained quantum sampler concentrates on the optimal set at the 150-qubit scale where greedy association starts to fail.

\section*{Acknowledgment}
We acknowledge the use of IBM Quantum services for this work. The views expressed are those of the authors, and do not reflect the official policy or position of IBM or the IBM Quantum team.

%=======================================================================
\bibliographystyle{IEEEtran}
\bibliography{main}

@inproceedings{hossain2026satlock,
  author    = {Hossain, Mohammad Arif and Hossain, Md Jafrin},
  title     = {{SATLOCK}: Handover-Coupled Scheduling for Weather-Resilient Quantum Key Distribution over {LEO} Constellations},
  booktitle = {Proceedings of the 3rd ACM SIGCOMM Workshop on Quantum Networks and Distributed Quantum Computing (QuNet '26)},
  series    = {QuNet '26},
  pages     = {35--41},
  year      = {2026},
  month     = aug,
  address   = {Denver, CO, USA},
  publisher = {ACM},
  isbn      = {979-8-4007-2892-1},
  doi       = {10.1145/3833409.3833444}
}

@ARTICLE{Liu_2026,
  author={Liu, Weiqi and Hossain, Mohammad Arif and Ansari, Nirwan},
  journal={IEEE Transactions on Network Science and Engineering}, 
  title={Multi-Critic Reinforcement Learning for Frame- and Tick-Rate Aware Satellite–Ground Integrated Heterogeneous Edge Networks}, 
  year={2026},
  volume={13},
  number={},
  pages={11239-11254},
  doi={10.1109/TNSE.2026.3728278}}

@IEEEtranBSTCTL{IEEEexample:BSTcontrol,
  CTLuse_forced_etal       = "yes",
  CTLmax_names_forced_etal = "6",
  CTLnames_show_etal       = "1"
}

@article{Yates2021,
  author = {Yates, Roy D. and Sun, Yin and Brown, D. Richard and Kaul, Sanjit K. and Modiano, Eytan and Ulukus, Sennur},
  title = {Age of Information: An Introduction and Survey},
  journal = {IEEE J. Sel. Areas Commun.},
  volume = {39},
  number = {5},
  pages = {1183--1210},
  month = may,
  year = {2021}
}

@article{Maatouk2020,
  author = {Maatouk, Ali and Kriouile, Saad and Assaad, Mohamad and Ephremides, Anthony},
  title = {The Age of Incorrect Information: A New Performance Metric for Status Updates},
  journal = {IEEE/ACM Trans. Netw.},
  volume = {28},
  number = {5},
  pages = {2215--2228},
  month = oct,
  year = {2020}
}

@article{Maatouk2023,
  author = {Maatouk, Ali and Assaad, Mohamad and Ephremides, Anthony},
  title = {The Age of Incorrect Information: An Enabler of Semantics-Empowered Communication},
  journal = {IEEE Trans. Wireless Commun.},
  volume = {22},
  number = {4},
  pages = {2621--2635},
  month = apr,
  year = {2023}
}

@inproceedings{Kriouile2021,
  author = {Kriouile, Saad and Assaad, Mohamad},
  title = {Minimizing the Age of Incorrect Information for Real-Time Tracking of {Markov} Remote Sources},
  booktitle = {Proc. IEEE Int. Symp. Inf. Theory (ISIT)},
  month = jul,
  year = {2021},
  pages = {2978--2983}
}

@book{BenTal2009,
  author = {Ben-Tal, Aharon and El Ghaoui, Laurent and Nemirovski, Arkadi},
  title = {Robust Optimization},
  publisher = {Princeton Univ. Press},
  address = {Princeton, NJ, USA},
  year = {2009}
}

@misc{Schulman2017,
  author = {Schulman, John and Wolski, Filip and Dhariwal, Prafulla and Radford, Alec and Klimov, Oleg},
  title = {Proximal Policy Optimization Algorithms},
  howpublished = {arXiv:1707.06347},
  year = {2017}
}

@misc{Farhi2014,
  author = {Farhi, Edward and Goldstone, Jeffrey and Gutmann, Sam},
  title = {A Quantum Approximate Optimization Algorithm},
  howpublished = {arXiv:1411.4028},
  year = {2014}
}

@article{Lucas2014,
  author = {Lucas, Andrew},
  title = {Ising Formulations of Many {NP} Problems},
  journal = {Front. Phys.},
  volume = {2},
  pages = {5},
  month = feb,
  year = {2014}
}

@article{Albash2018,
  author = {Albash, Tameem and Lidar, Daniel A.},
  title = {Adiabatic Quantum Computation},
  journal = {Rev. Mod. Phys.},
  volume = {90},
  number = {1},
  pages = {015002},
  month = jan,
  year = {2018}
}

@article{Boyd2011,
  author = {Boyd, Stephen and Parikh, Neal and Chu, Eric and Peleato, Borja and Eckstein, Jonathan},
  title = {Distributed Optimization and Statistical Learning via the Alternating Direction Method of Multipliers},
  journal = {Found. Trends Mach. Learn.},
  volume = {3},
  number = {1},
  pages = {1--122},
  year = {2011}
}

@article{Saad2020,
  author = {Saad, Walid and Bennis, Mehdi and Chen, Mingzhe},
  title = {A Vision of {6G} Wireless Systems: Applications, Trends, Technologies, and Open Research Problems},
  journal = {IEEE Netw.},
  volume = {34},
  number = {3},
  pages = {134--142},
  month = may,
  year = {2020}
}

@article{Giordani2021,
  author = {Giordani, Marco and Zorzi, Michele},
  title = {Non-Terrestrial Networks in the {6G} Era: Challenges and Opportunities},
  journal = {IEEE Netw.},
  volume = {35},
  number = {2},
  pages = {244--251},
  month = mar,
  year = {2021}
}

@article{Kato2019,
  author = {Kato, Nei and Fadlullah, Zubair Md. and Tang, Fengxiao and Mao, Bomin and Tani, Shigenori and Okamura, Atsushi and Liu, Jiajia},
  title = {Optimizing Space-Air-Ground Integrated Networks by Artificial Intelligence},
  journal = {IEEE Wireless Commun.},
  volume = {26},
  number = {4},
  pages = {140--147},
  month = aug,
  year = {2019}
}

@article{Kadota2018,
  author = {Kadota, Igor and Sinha, Abhishek and Uysal-Biyikoglu, Elif and Singh, Rahul and Modiano, Eytan},
  title = {Scheduling Policies for Minimizing Age of Information in Broadcast Wireless Networks},
  journal = {IEEE/ACM Trans. Netw.},
  volume = {26},
  number = {6},
  pages = {2637--2650},
  month = dec,
  year = {2018}
}

@inproceedings{Chen2021,
  author = {Chen, Yutao and Ephremides, Anthony},
  title = {Minimizing Age of Incorrect Information for Unreliable Channel With Power Constraint},
  booktitle = {Proc. IEEE Global Commun. Conf. (GLOBECOM)},
  month = dec,
  year = {2021},
  pages = {1--6}
}

@techreport{TR38811,
  author      = {{3GPP}},
  title       = {Study on {New Radio} ({NR}) to support non-terrestrial networks},
  institution = {3rd Generation Partnership Project},
  number      = {TR 38.811, V15.4.0},
  year        = {2020},
  month       = oct
}

@inproceedings{Vallado2006,
  author    = {Vallado, David A. and Crawford, Paul and Hujsak, Richard and Kelso, T. S.},
  title     = {Revisiting {Spacetrack Report \#3}},
  booktitle = {Proc. AIAA/AAS Astrodynamics Specialist Conf.},
  year      = {2006},
  address   = {Keystone, CO},
  note      = {{AIAA} 2006-6753}
}

@article{Hadfield2019,
  author  = {Hadfield, Stuart and Wang, Zhihui and O'Gorman, Bryan and Rieffel, Eleanor G. and Venturelli, Davide and Biswas, Rupak},
  title   = {From the Quantum Approximate Optimization Algorithm to a Quantum Alternating Operator Ansatz},
  journal = {Algorithms},
  volume  = {12},
  number  = {2},
  pages   = {34},
  year    = {2019}
}

@misc{Brandao2018,
  author = {Brand{\~a}o, Fernando G. S. L. and Broughton, Michael and Farhi, Edward and Gutmann, Sam and Neven, Hartmut},
  title  = {For fixed control parameters the quantum approximate optimization algorithm's objective function value concentrates for typical instances},
  year   = {2018},
  note   = {arXiv:1812.04170}
}

@article{Whittle1988,
  author  = {Whittle, Peter},
  title   = {Restless bandits: Activity allocation in a changing world},
  journal = {J. Appl. Probab.},
  volume  = {25},
  number  = {A},
  pages   = {287--298},
  year    = {1988}
}

@article{Viola1999,
  author  = {Viola, Lorenza and Knill, Emanuel and Lloyd, Seth},
  title   = {Dynamical decoupling of open quantum systems},
  journal = {Phys. Rev. Lett.},
  volume  = {82},
  number  = {12},
  pages   = {2417--2421},
  year    = {1999}
}

@article{Wallman2016,
  author  = {Wallman, Joel J. and Emerson, Joseph},
  title   = {Noise tailoring for scalable quantum computation via randomized compiling},
  journal = {Phys. Rev. A},
  volume  = {94},
  number  = {5},
  pages   = {052325},
  year    = {2016}
}

@article{HiGHS2018,
  author  = {Huangfu, Qi and Hall, J. A. J.},
  title   = {Parallelizing the dual revised simplex method},
  journal = {Math. Program. Comput.},
  volume  = {10},
  number  = {1},
  pages   = {119--142},
  year    = {2018}
}

@misc{Qiskit2024,
  author = {Javadi-Abhari, Ali and Treinish, Matthew and Krsulich, Kevin and Wood, Christopher J. and Lishman, Jake and Gacon, Julien and Martiel, Simon and Nation, Paul D. and Bishop, Lev S. and Cross, Andrew W. and Johnson, Blake R. and Gambetta, Jay M.},
  title  = {Quantum computing with {Qiskit}},
  year   = {2024},
  note   = {arXiv:2405.08810}
}

@article{Arora2012,
  author  = {Arora, Sanjeev and Hazan, Elad and Kale, Satyen},
  title   = {The multiplicative weights update method: A meta-algorithm and applications},
  journal = {Theory Comput.},
  volume  = {8},
  number  = {6},
  pages   = {121--164},
  year    = {2012}
}

@article{Hong2024,
  author  = {Hong, Hui and Jiao, Jian and Yang, Tao and Wang, Ye and Lu, Rongxing and Zhang, Qinyu},
  title   = {Age of Incorrect Information Minimization for Semantic-Empowered {NOMA} System in {S-IoT}},
  journal = {IEEE Trans. Wireless Commun.},
  volume  = {23},
  number  = {6},
  pages   = {6639--6652},
  month   = jun,
  year    = {2024}
}

@misc{Ayik2023,
  author       = {Ayik, Muratcan and Ceran, Elif Tugce and Uysal, Elif},
  title        = {Optimization of {AoII} and {QAoII} in Multi-User Links},
  howpublished = {arXiv:2305.00191},
  year         = {2023}
}

@article{Liao2024,
  author  = {Liao, Mingjun and Wang, Ruyan and Zhang, Puning and Xian, Ziyun},
  title   = {Information Freshness Optimal Resource Allocation for {LEO}-Satellite {Internet of Things}},
  journal = {IEEE Internet Things J.},
  volume  = {11},
  number  = {10},
  pages   = {17372--17387},
  month   = may,
  year    = {2024}
}

@misc{Lang2025,
  author       = {Lang, Zifan and Liu, Guixia and Sun, Geng and Li, Jiahui and Wang, Jiacheng and Yuan, Weijie and Niyato, Dusit and Kim, Dong In},
  title        = {Joint {AoI} and Handover Optimization in Space-Air-Ground Integrated Network},
  howpublished = {arXiv:2509.12716},
  year         = {2025}
}

@article{Chen2021DRO,
  author  = {Chen, Yali and Ai, Bo and Niu, Yong and Zhang, Hongliang and Han, Zhu},
  title   = {Energy-Constrained Computation Offloading in Space-Air-Ground Integrated Networks Using Distributionally Robust Optimization},
  journal = {IEEE Trans. Veh. Technol.},
  volume  = {70},
  number  = {11},
  pages   = {12113--12125},
  month   = nov,
  year    = {2021}
}

@article{Wang2020,
  author  = {Wang, Zhihui and Rubin, Nicholas C. and Dominy, Jason M. and Rieffel, Eleanor G.},
  title   = {{$XY$} Mixers: Analytical and Numerical Results for the Quantum Alternating Operator Ansatz},
  journal = {Phys. Rev. A},
  volume  = {101},
  number  = {1},
  pages   = {012320},
  month   = jan,
  year    = {2020}
}

@article{Bucher2026,
  author  = {Bucher, David and Janetschek, Maximilian and Poppel, Michael and Stein, Jonas and Linnhoff-Popien, Claudia and Feld, Sebastian},
  title   = {Constrained Quantum Optimization via Iterative Warm-Start {$XY$}-Mixers},
  journal = {New J. Phys.},
  volume  = {28},
  number  = {8},
  pages   = {084501},
  month   = aug,
  year    = {2026}
}

@misc{Yan2026,
  author       = {Yan, Qiben and Stenger, John P. T. and Gunlycke, Daniel},
  title        = {Toward Quantum-Optimized Flow Scheduling in Multi-Beam Digital Satellites},
  howpublished = {arXiv:2603.00701},
  year         = {2026}
}

@article{Gambella2020,
  author  = {Gambella, Claudio and Simonetto, Andrea},
  title   = {Multiblock {ADMM} Heuristics for Mixed-Binary Optimization on Classical and Quantum Computers},
  journal = {IEEE Trans. Quantum Eng.},
  volume  = {1},
  pages   = {1--22},
  year    = {2020}
}

@inproceedings{Halffmann2024,
  author    = {Halffmann, Pascal and Trebing, Michael and Lenk, Steve},
  title     = {Harnessing Inferior Solutions for Superior Outcomes: Obtaining Robust Solutions From Quantum Algorithms},
  booktitle = {Proc. Genetic Evol. Comput. Conf. Companion (GECCO)},
  month     = jul,
  year      = {2024},
  pages     = {1950--1953}
}

\end{document}